\documentclass[10pt,journal]{IEEEtran}
\usepackage[dvips]{graphicx}
\usepackage{url}
\usepackage[cmex10]{amsmath}
\usepackage{multirow}
\usepackage{epsfig}
\usepackage{mathrsfs}
\usepackage{amssymb}
\usepackage{comment}
\usepackage{bm}
\usepackage{array}
\usepackage{amsthm}
\usepackage{blkarray}
\usepackage{enumerate}
\usepackage{chemarrow}
\usepackage{cite}
\usepackage[lined,ruled,linesnumbered]{algorithm2e}
\usepackage{epsf,psfrag}
\usepackage{color}
\usepackage[tight,footnotesize]{subfigure}

\makeatletter
\newif\if@restonecol
\makeatother

\usepackage[lined,ruled,linesnumbered]{algorithm2e}

\def\clr{\color{black}}
\def\clrr{\color{black}}

\makeatletter
\def\ps@headings{%
\def\@oddhead{\mbox{}\scriptsize\rightmark \hfil \thepage}%
\def\@evenhead{\scriptsize\thepage \hfil \leftmark\mbox{}}%
\def\@oddfoot{}%
\def\@evenfoot{}}
\makeatother
\makeatletter
\def\old@comma{,}
\catcode`\,=13
\def,{%
  \ifmmode%
    \old@comma\discretionary{}{}{}%
  \else%
    \old@comma%
  \fi%
}
\makeatother

\newtheorem{theorem}{Theorem}
\newtheorem{lemma}[theorem]{Lemma}
\newtheorem{proposition}[theorem]{Proposition}
\newtheorem{corollary}[theorem]{Corollary}
\newtheorem{claim}[theorem]{Claim}

\newtheorem{definition}[theorem]{Definition}

\newcommand{\OmegaCAP}{\Omega^{\mbox{\tiny CAP}}}
\newcommand{\OmegaCSP}{\Omega^{\mbox{\tiny CSP}}}
\newcommand{\OmegaUP}{\Omega^{\mbox{\tiny UP}}}
\newcommand{\deltaMin}{\delta_{\min}} 
\newcommand{\MSC}{\mbox{MSC}}
\newcommand{\GSC}{\mbox{GSC}}

\newcommand{\SCAP}{S^*_{\mbox{\tiny CAP}}}
\newcommand{\SCSP}{S^*_{\mbox{\tiny CSP}}}
\newcommand{\SUP}{S^*_{\mbox{\tiny UP}}}

\newcommand{\SCSPi}{S_{\mbox{\tiny CSP}}^{\mbox{\tiny inner}}}
\newcommand{\SCSPo}{S_{\mbox{\tiny CSP}}^{\mbox{\tiny outer}}}
\newcommand{\SUPi}{S_{\mbox{\tiny UP}}^{\mbox{\tiny inner}}}
\newcommand{\SUPo}{S_{\mbox{\tiny UP}}^{\mbox{\tiny outer}}}
\newcommand{\Sinner}{S^{\mbox{\tiny inner}}}
\newcommand{\Souter}{S^{\mbox{\tiny outer}}}

\DeclareMathAlphabet{\mathcal}{OMS}{cmsy}{m}{n}

\ifCLASSINFOpdf

\else

\fi

\begin{document}

\IEEEoverridecommandlockouts

\title{Network Capability in Localizing Node Failures via End-to-end Path Measurements}
\author{\IEEEauthorblockN{Liang Ma\IEEEauthorrefmark{2}, Ting He\IEEEauthorrefmark{2}, Ananthram Swami\IEEEauthorrefmark{4}, Don Towsley\IEEEauthorrefmark{1}, and Kin K. Leung\IEEEauthorrefmark{3}\\}
\IEEEauthorblockA{\IEEEauthorrefmark{2}IBM T. J. Watson Research Center, Yorktown, NY, USA. Email: \{maliang, the\}@us.ibm.com\\
\IEEEauthorrefmark{4}Army Research Laboratory, Adelphi, MD, USA. Email: ananthram.swami.civ@mail.mil\\
\IEEEauthorrefmark{1}University of Massachusetts, Amherst, MA, USA. Email: towsley@cs.umass.edu\\
\IEEEauthorrefmark{3}Imperial College, London, UK. Email: kin.leung@imperial.ac.uk
}

\thanks{Research was sponsored by the U.S. Army Research Laboratory and the U.K. Ministry of Defence and was accomplished under Agreement Number W911NF-06-3-0001. The views and conclusions contained in this document are those of the authors and should not be interpreted as representing the official policies, either expressed or implied, of the U.S. Army Research Laboratory, the U.S. Government, the U.K. Ministry of Defence or the U.K. Government. The U.S. and U.K. Governments are authorized to reproduce and distribute reprints for Government purposes notwithstanding any copyright notation hereon.

}
}

\maketitle
\thispagestyle{empty}
\begin{abstract}
We investigate the capability of localizing node failures in communication networks from binary states (normal/failed) of end-to-end paths. Given a set of nodes of interest, uniquely localizing failures within this set requires that different observable path states associate with different node failure events. However, this condition is difficult to test on large networks due to the need to enumerate all possible node failures. Our first contribution is a set of sufficient/necessary conditions for identifying a bounded number of failures within an arbitrary node set that can be tested in polynomial time. In addition to network topology and locations of monitors, our conditions also incorporate constraints imposed by the probing mechanism used. We consider three probing mechanisms that differ according to whether measurement paths are (i) arbitrarily controllable, (ii) controllable but cycle-free, or (iii) uncontrollable (determined by the default routing protocol). Our second contribution is to quantify the capability of failure localization through (1) the maximum number of failures (anywhere in the network) such that failures within a given node set can be uniquely localized, and (2) the largest node set within which failures can be uniquely localized under a given bound on the total number of failures. Both measures in (1--2) can be converted into functions of a per-node property, which can be computed efficiently based on the above sufficient/necessary conditions. We demonstrate how measures (1--2) proposed for quantifying failure localization capability can be used to evaluate the impact of various parameters, including topology, number of monitors, and probing mechanisms.

\end{abstract}
\begin{IEEEkeywords}
Network Tomography, Failure Localization, Identifiability Condition, Maximum Identifiability Index
\end{IEEEkeywords}

\IEEEpeerreviewmaketitle

\section{Introduction}
\label{intro}

Effective monitoring of network performance is essential for network operators in building reliable communication networks that are robust to service disruptions. In order to achieve this goal, the monitoring infrastructure must be able to detect network misbehaviors (e.g., unusually high loss/latency, unreachability) and localize the sources of the anomaly (e.g., malfunction of certain routers) in an accurate and timely manner. Knowledge of where problematic network elements reside in the network is particularly useful for fast service recovery, e.g., the network operator can migrate affected services and/or reroute traffic. However, localizing network elements that cause a service disruption can be challenging. The straightforward approach of directly monitoring the health of individual elements is not always feasible due to traffic overhead, access control, or lack of protocol support at internal nodes. Moreover, built-in monitoring agents running on network elements cannot detect problems caused by misconfigured/unanticipated interactions between network layers, where end-to-end communication is disrupted but individual network elements along the path remain functional (a.k.a. \emph{silent failures}) \cite{Kompella07infocom}. These limitations call for a {\clr different} approach that can diagnose the health of network elements from the health of end-to-end communications perceived between measurement points.\looseness=-1

One such approach, generally known as \emph{network tomography} \cite{Coates02}, {\clr focuses on inferring internal network characteristics based on end-to-end performance measurements} from a subset of nodes with monitoring capabilities, referred to as \emph{monitors}.  Unlike direct measurement, network tomography only relies on end-to-end performance (e.g., path connectivity) experienced by data packets, thus addressing issues such as overhead, lack of protocol support, and silent failures. In cases where the network characteristic of interest is binary (e.g., \emph{normal} or \emph{failed}), this approach is known as \emph{Boolean network tomography} \cite{Ghita11conext}.

In this paper, we study an application of Boolean network tomography to localize node failures from measurements of path states\footnote{This model can also capture link failures by transforming the topology into a logical topology with each link represented by a virtual node connected to the nodes incident to the link.}. Under the assumption that a measurement path is normal if and only if all nodes on this path behave normally, we formulate the problem as a system of Boolean equations, where the unknown variables are the binary node states, and the known constants are the observed states of measurement paths. The goal of Boolean network tomography is essentially to solve this system of Boolean equations.

Because the observations are coarse-grained (path normal/failed), it is usually impossible to uniquely identify node states from path measurements. For example, if two nodes always appear together in measurement paths, then upon observing failures of all these paths, we can at most deduce that one of these nodes (or both) has failed but cannot determine which one. Because there are often multiple explanations for given path failures, existing work mostly focuses on finding the minimum set of failed nodes that most probably involves failed nodes. Such an approach, however, does not guarantee that nodes in this minimum set have failed or that nodes outside the set have not. Generally, to distinguish between two possible failure sets, there must exist a measurement path that traverses one and only one of these two sets. There is, however, a lack of understanding of what this requires in terms of observable network properties such as topology, monitor placement, and measurement routing. On the other hand, even if there exists ambiguity in failure localization across the entire network, it is still possible to uniquely localize node failures in a specific sub-network (e.g., sub-network with a large fraction of monitors). To determine such unique failure localization in sub-networks, we need to understand how it is related to network properties.\looseness=-1

In this paper, we consider three closely related problems: Let $S$ denote a set of nodes of interest (i.e., there can be ambiguity in determining the states of nodes outside $S$; however, the states of nodes in $S$ must be uniquely determinable). (1) If the number of simultaneous node failures is bounded by $k$, then under what conditions can one uniquely localize failed nodes in $S$ from path measurements available in the entire network? (2) What is the maximum number of simultaneous node failures (i.e., the largest value of $k$) such that any failures within $S$ can be uniquely localized? (3) What is the largest node set within which failures can be uniquely localized, if the total number of failures is bounded by $k$? Answers to questions (2) and (3) together quantify a network's capability to localize failures from end-to-end measurements: question (2) characterizes the \emph{scale} of failures and question (3) the \emph{scope} of localization. Clearly, answers to the above questions depend on which paths are measurable, which in turn depends on network topology, placement of monitors, and the routing mechanism of probes. We will study all these problems in the context of the following classes of probing mechanisms: (i) \emph{Controllable Arbitrary-path Probing (CAP)}, where any measurement path can be set up by monitors, (ii) \emph{Controllable Simple-path Probing (CSP)}, where any measurement path can be set up, provided it is cycle-free, and (iii) \emph{Uncontrollable Probing (UP)}, where measurement paths are determined by the default routing protocol. These probing mechanisms assume different levels of control over routing of probing packets and are feasible in different network scenarios (see Section~\ref{subsec:Classification of Probing Mechanisms}); answers to the above three problems under these probing mechanisms thus provide insights on how the level of control bestowed on the monitoring system affects its capability in failure localization.\looseness=-1

\subsection{Related Work}
\label{subsec:related_work}

Existing work can be broadly classified into single failure localization and multiple failure localization. Single failure localization assumes that multiple simultaneous failures happen with negligible probability. Under this assumption, \cite{Bejerano03INFOCOM,Horton03} propose efficient algorithms for monitor placement such that any single failure can be detected and localized. To improve the resolution in characterizing failures, range tomography in \cite{Zarifzadeh12IMC} not only localizes the failure, but also estimates its severity (e.g., congestion level). These works, however, ignore the fact that multiple failures occur more frequently than one may imagine \cite{Markopoulou04infocom}. In this paper, we consider the general case of localizing multiple failures.

Multiple failure localization faces inherent uncertainty. Most existing works address this uncertainty by attempting to find the minimum set of network elements whose failures explain the observed path states. Under the assumption that failures are low-probability events, this approach generates the most probable failure set among all possibilities. Using this approach, \cite{Duffield03,Duffield06TInfo} propose solutions for networks with tree topologies, which are later extended to general topologies in \cite{Kompella07infocom}. {\clr Similarly, \cite{Zeng12CoNEXT} proposes to localize link failures by minimizing false positives; however, it cannot guarantee unique failure localization.} In a Bayesian formulation, \cite{Nguyen07infocom} proposes a two-stage solution which first estimates the failure (loss rate above threshold) probabilities of different links and then infers the most likely failure set for subsequent measurements.
By augmenting path measurements with (partially) available control plane information (e.g., routing messages), \cite{Dhamdhere07CoNEXT,Huang08ccr} propose a greedy heuristic for troubleshooting network unreachability in multi-AS (Autonomous System) networks that has better accuracy than benchmarks using only path measurements.

Little is known when we insist on \emph{uniquely} localizing network failures. Given a set of monitors known to uniquely localize failures on paths between themselves, \cite{Nguyen04PAM} develops an algorithm to remove redundant monitors such that all failures remain identifiable.
If the number of failed links is upper bounded by $k$ and the monitors can probe arbitrary cycles or paths containing cycles, \cite{Ahuja08} proves that the network must be $(k+2)$-edge-connected to identify any failures up to $k$ links using one monitor, which is then used to derive requirements on monitor placement for general topologies. Solving node failure localization using the results of \cite{Ahuja08}, however, requires a topology transformation that maps each node to a link while maintaining adjacency between nodes and feasibility of measurement paths. To our knowledge, no such transformation exists whose output satisfies the assumptions of \cite{Ahuja08} (undirected graph, measurement paths not containing repeated links).
Later, \cite{Cho14Elsevier} proves that under a CAP-like probing mechanism, the condition can be relaxed to the network being $k$-edge-connected. Both \cite{Ahuja08,Cho14Elsevier} focus on placing monitors and constructing measurement paths to localize a given number of failures; in contrast, we focus on characterizing the capability of failure localization under a given monitor placement and constraints on measurement paths. In previous work \cite{Ma14IMC}, we propose efficient testing conditions and algorithms to quantify the capability of localizing node failures in the entire network; however, we did not consider the case that even if some node states cannot be uniquely determined, we may still be able to unambiguously determine the states of some other nodes. In this paper, we thus investigate the relationships between the capability of localizing node failures and explicit network properties such as topology, placement of monitors, probing mechanism, and nodes of interest, with focus on developing efficient algorithms to characterize the capability under given settings.\looseness=-1

A related but fundamentally different line of work is graph-constrained group testing \cite{Cheraghchi12}, which studies the minimum number of measurement paths needed to uniquely localize a given number of (node/link) failures, using a CAP-like probing mechanism. In contrast, we seek to characterize the type of failures (number and location) that can be uniquely localized using a variety of probing mechanisms.

\subsection{Summary of Contributions}

We study the fundamental capability of a network with {\clrr arbitrarily} placed monitors to uniquely localize node failures from binary end-to-end measurements between monitors.
Our contributions are five-fold:

\emph{1)} We propose two novel measures to quantify the capability of failure localization, (i) \emph{maximum identifiability index} of a given node set, which characterizes the maximum number of simultaneous failures such that failures within this set can be uniquely localized, and (ii) \emph{maximum identifiable set} for a given upper bound on the number of simultaneous failures, which represents the largest node set within which failures can be uniquely localized if the failure event satisfies the bound. We show that both measures can be expressed as functions of per-node maximum identifiability index (i.e., maximum number of failures such that the failure of a given node can be uniquely determined).

\emph{2)} We establish necessary/sufficient conditions for uniquely localizing failures in a given set under a bound on the total number of failures, which are applicable to all probing mechanisms. We then convert these conditions into more concrete conditions in terms of network topology and placement of monitors, under the three different probing mechanisms (CAP, CSP, and UP), which can be tested in polynomial time.

\emph{3)} We show that a special relationship between the above necessary/sufficient conditions leads to tight upper/lower bounds on the maximum identifiability index of a given set that narrows its value to at most two consecutive integers. These conditions also enable a strategy for constructing inner/outer bounds (i.e., subset/superset) of the maximum identifiable set. These bounds are polynomial-time computable under CAP and CSP. While they are NP-hard to compute under UP, we present a greedy heuristic to compute a pair of relaxed bounds that frequently coincide with the original bounds in practice.\looseness=-1

\emph{4)} We evaluate the proposed measures under different probing mechanisms on random and real topologies. Our evaluation shows that controllable probing, especially CAP, significantly improves the capability of node failure localization over uncontrollable probing. Our result also reveals novel insights into the distribution of per-node maximum identifiability index and its relationship with graph-theoretic node properties.

\emph{Note: } Our results are also applicable to transient failures as long as node failures persist during probing (i.e., leading to failures of all traversing paths). We have limited our observations to binary states (normal/failed) of measurement paths. It is possible in some networks to obtain extra information from probes, e.g., rerouted paths after a default path fails, in which case our solution provides lower bounds on the capability of localizing failures. 
{\clr Furthermore, we do not make any assumption on the distribution or correlation of node failures across the network. In some application scenarios (e.g., datacenter networks), node failures may be correlated (e.g., all routers sharing the same power/chiller). We leave the characterization of failure localization in the presence of such additional information to future work.}

\vspace{.5em}
The rest of the paper is organized as follows. Section~\ref{Sect:ProblemFormulation} formulates the problem.
Section~\ref{sec:TheoryFoundation} presents the theoretical foundations for identifying node failures, followed by verifiable identifiability conditions for specific classes of probing mechanisms in Section~\ref{sec:Verifiable Identifiability Conditions}. Based on the derived conditions, tight bounds on the maximum identifiability index are presented in Section~\ref{sec:Characterization of Maximum Identifiability}, and inner/outer bounds on the maximum identifiable set are established in Section~\ref{sec:CharacterizationMaximumIdentifiableSet}. We evaluate the established bounds on various synthetic/real topologies in Section~\ref{sec:Impact of Probing Mechanisms} to study the impact of various parameters (topology, number of monitors, probing mechanism) on the capability of node failure localization. Finally, Section~\ref{sec:Conclusion} concludes the paper.\looseness=-1

\section{Problem Formulation}
\label{Sect:ProblemFormulation}

\subsection{Models and Assumptions}\label{subsec:Models and Assumptions}

\begin{table}[tb]
\vspace{-.0em}
\small
\renewcommand{\arraystretch}{1.3}
\caption{Graph-related Notations} \label{t notion}
\vspace{.0em}
\centering
\begin{tabular}{r|m{6.1cm}}
  \hline
  \textbf{Symbol} & \textbf{Meaning} \\
  \hline
  $V$, $L$ & set of nodes/links ($\xi:=|L|$)\\
  \hline
  $M,\: N$ & set of monitors/non-monitors ($M\cup N = V$, $\mu:=|M|,\: \sigma:=|N|$) \\
  \hline
  $k$ & maximum number of simultaneous non-monitor failures\\
  \hline
  $V(\mathcal{G})$ & set of nodes in $\mathcal{G}$ \\
  \hline
  $\mathcal{N}(M)$ & set of non-monitors that are neighbors of at least one monitor in $M$ ($\theta:=|\mathcal{N}(M)|$)\\
  \hline
  $\mathcal{L}(V,\: W)$ & $\mathcal{L}(V,\: W)= \{\mbox{link }vw:\: \forall v\in V,\: w\in W,\: v\neq w\}$\\
  \hline
  $\mathcal{G}-L'$ & delete links: $\mathcal{G}-L'=(V,L\setminus L')$, where ``$\setminus$'' is setminus\\
  \hline
  $\mathcal{G}+L'$ & add links: $\mathcal{G}+L'=(V,L\cup L')$, where the end-points of links in $L'$ must be in $V$\\
  \hline
  $\mathcal{G}-V'$ & delete nodes: $\mathcal{G}-V'=(V\setminus V',L\setminus L(V'))$, where $L(V')$ is the set of links incident to nodes in $V'$\\
  \hline
  $\mathcal{G}+V'$ & add nodes: $\mathcal{G}+V'=(V\cup V',L)$ \\
  \hline
  $\mathcal{G}^*$ & auxiliary graph of $\mathcal{G}$ (see Fig.~\ref{Fig:AuxiliaryGraph})\\
  \hline
  $\mathcal{G}_m$ & auxiliary graph of $\mathcal{G}$ w.r.t. monitor $m$ (see Fig.~\ref{Fig:AuxiliaryGraph})\\
  \hline
  $\mathcal{G}'$ & extended graph of $\mathcal{G}$ (see Fig.~\ref{fig:extended_graph})\\
  \hline
  $\Omega(S)$, $\Omega(v)$ & maximum identifiability index of $S$ or $v$ ($S$: a set of nodes, $v$: a node)\\
  \hline
  $S^*(k)$ & maximum $k$-identifiable set\\
  \hline
  $S^{\mbox{\tiny inner}}(k)$ & subset of $S^*(k)$\\
  \hline
  $S^{\mbox{\tiny outer}}(k)$ & superset of $S^*(k)$\\
  \hline
\end{tabular}
\vspace{-1em}
\end{table}
\normalsize

We assume that the network topology is known and model it as an undirected
graph\footnote{We use the terms \emph{network} and \emph{graph} interchangeably.}  $\mathcal{G}=(V,L)$, where $V$ and $L$ are the sets of nodes and links. In $\mathcal{G}$, the number of neighbors of node $v$ is called the \emph{degree} of {\clr $v$; $\xi:=|L|$ denotes the number of links. Note} that graph $\mathcal{G}$ can represent a logical topology where each node in $\mathcal{G}$ corresponds to a physical subnetwork. Without loss of generality, we assume $\mathcal{G}$ is connected, as different connected components have to be monitored separately. 

A subset of nodes $M$ ($M\subseteq V$) are \emph{monitors} that can initiate and collect measurements. The rest of the nodes, denoted by $N:=V\setminus M$, are \emph{non-monitors}. Let $\mu:=|M|$ and $\sigma:=|N|$ denote the numbers of monitors and non-monitors. We assume that monitors do not fail during the measurement process, as failed monitors can be directly detected and excluded (assuming centralized control within the monitoring system). Non-monitors, on the other hand, can fail, and a failure event may involve simultaneous failures of multiple non-monitors. Depending on the adopted probing mechanism, monitors measure the states of nodes by sending probes along certain paths. Let $P$ denote the set of all \emph{possible measurement paths}; for given $\mathcal{G}$ and $M$, different probing mechanisms can lead to different sets of measurement paths, which will be specified later. We use \emph{node state} (\emph{path state}) to refer to the binary state, failed or normal, of a node (path), where a path fails if and only if at least one node on the path fails.
Table~\ref{t notion} summarizes graph-related notations used in this paper.

Let $\textbf{w}=(W_{1},\ldots,W_{\sigma})^T$ be the binary column vector of the states of all non-monitors and $\textbf{c}=(C_{1},\ldots,C_{\gamma})^T$ the binary column vector ($\gamma=|P|$) of the states of all measurement paths. For both node and path states, $0$ represents ``normal'' and $1$ represents ``failed''. We relate the path states to the node states through the following Boolean linear system:
\begin{equation}\label{Eq:problem formulation}
    \mathbf{R}\odot\textbf{w}=\textbf{c},
\end{equation}
where $\mathbf{R}=(R_{ij})$ is a $\gamma\times \sigma$ \emph{measurement matrix}, with each entry $R_{ij}\in \{0,\: 1\}$ denoting whether non-monitor $v_j$ is present on path $\mathcal{P}_i$ ($1$: yes, $0$: no), and ``$\odot$'' is the Boolean matrix product, i.e., $C_{i}=\vee^\sigma_{j=1} (R_{ij}\wedge W_{j})$. The goal of Boolean network tomography is to invert this Boolean linear system to solve for all/part of the elements in \textbf{w} given \textbf{R} and \textbf{c}. Intuitively, for a node set $S$ ($S\subseteq N$), any node failures in $S$ are identifiable if and only if the corresponding states of $S$ in $\textbf{w}$ are uniquely determinable by (\ref{Eq:problem formulation}).

\subsection{Definitions}

{\clr Let {\clrr a} \emph{failure set} $F$ be a set of non-monitors ($F\subseteq N$) that fail simultaneously.} Note that the collection of all failure sets in a given network covers all possible failure scenarios (each corresponds to a failure set) that can occur in this network; the goal of failure localization is to infer the current failure set from the states of measurement paths. The challenge for this problem is that there may exist multiple failure sets leading to the same path states, causing ambiguity. Let $P_{F}$ denote the set of all measurement paths affected by a failure set ${F}$ (i.e., paths traversing at least one node in ${F}$). To quantify the capability of uniquely determining the failure set, we introduce the following definitions.

\begin{definition}\label{def:distinguishability}
Given a network $\mathcal{G}$ and a set of measurement paths $P$, two failure sets ${F}_1$ and ${F}_2$ are \emph{distinguishable} if and only if $P_{{F}_1}\neq P_{{F}_2}$, i.e., $\exists$ a path that traverses one and only one of ${F}_1$ and ${F}_2$.
\end{definition}

Definition~\ref{def:distinguishability} implies that two potential failure sets must be associated with different observable path states for monitors to determine which set of nodes have failed. While uniquely localizing arbitrary failures requires all subsets of $N$ to be pairwise distinguishable, we can relax this requirement by only considering failure sets of size bounded by $k$ ($k\geq 1$), which represents the scale of probable failure events. Moreover, in practice, we are usually interested in the states of a subset of nodes $S$ ($S\subseteq N$), in which case the goal is to only ensure unique failure localization within $S$. Note that failures ($F$) may occur anywhere in the network ($F\subseteq N$) and are not restricted to $S$.

\begin{definition}\label{def:identifiability}
Given a network $\mathcal{G}$ (with non-monitor set $N$) and a node set $S$ of interest ($S\subseteq N$):
\begin{enumerate}
  \item $S$ is \emph{$k$-identifiable} if for any two failure sets $F_1$ and $F_2$ satisfying (1) $|F_i|\leq k$ ($i=1,2$) and (2) $F_1\cap S \neq F_2 \cap S$, $F_1$ and $F_2$ are distinguishable.
  \item The \emph{maximum identifiability index of $S$}, denoted by $\Omega(S)$, is the maximum value of $k$ such that $S$ is $k$-identifiable.
\end{enumerate}
\end{definition}

Intuitively, if a node set $S$ is $k$-identifiable, then the states (normal/failed) of all nodes within this set are unambiguously determinable from the observed path states, provided the total number of failures (anywhere in the network) is bounded by $k$. The maximum identifiability index $\Omega(S)$ characterizes the network's capability to uniquely localize failures in $S$. Definition~\ref{def:identifiability} generalizes the notion of network-wide $k$-identifiability and maximum identifiability index introduced in \cite{Ma14IMC}, where only the case of $S=N$ was considered.
In the special case of $S=\{v\}$, we say that node $v$ is $k$-identifiable; the maximum identifiability index of $S=\{v\}$ is denoted by $\Omega(v)$. Note that the subset of a $k$-identifiable set is also $k$-identifiable. We are therefore interested in the maximum such set.

\begin{definition}\label{def:max-identifiableSet}
Given $k$, the maximum $k$-identifiable set, denoted by $S^*(k)$, is the largest-cardinality non-monitor set that is $k$-identifiable.
\end{definition}

According to Definition~\ref{def:max-identifiableSet}, it seems that the maximum $k$-identifiable set is defined based on its cardinality, and thus may not be unique. Nevertheless, we prove in Section~\ref{sec:setProperty} that $S^*(k)$ is unique. The significance of the maximum $k$-identifiable set is that it measures the completeness of the inferred network state: it contains all nodes whose states can be inferred reliably from the observed path states, as long as the total number of failures in the network is bounded by $k$. Note that $k$ is a design parameter capturing the scale of failures that the system is designed to handle.

\subsection{Classification of Probing Mechanisms}\label{subsec:Classification of Probing Mechanisms}

The above definitions are all defined with respect to (w.r.t.) a given set of measurement paths $P$. Given the topology $\mathcal{G}$ and monitor locations $M$, the probing mechanism plays a crucial role in determining $P$. Depending on the flexibility of probing and the cost of deployment, we classify probing mechanisms into one of three classes:
\begin{enumerate}
\item \emph{Controllable Arbitrary-path Probing (CAP):} $P$ includes any path/cycle, allowing  repeated nodes/links, provided each path/cycle starts and ends at monitors.
\item \emph{Controllable Simple-path Probing (CSP):} $P$ includes any \emph{simple} path between distinct monitors, not including repeated nodes.
\item \emph{Uncontrollable Probing (UP):} $P$ is the set of paths between monitors determined by the routing protocol used by the network, not controllable by the monitors.
\end{enumerate}

Although CAP allows probes to traverse each node/link an arbitrary number of times, it suffices to consider paths where each probe traverses each link at most once in either direction for the sake of localizing node failures.

{\clr These probing mechanisms clearly provide decreasing flexibility to the monitors and therefore  decreasing capability to localize failures. However, they also offer decreasing deployment cost. CAP represents the most flexible probing mechanism and provides an upper bound on failure localization capability. In traditional networks, CAP is feasible at the IP layer if \emph{(strict) source routing} (an IP option) \cite{RFCsourceRouting} is enabled at all nodes\footnote{Source routing allows nodes to modify the source and the destination addresses in packet headers hop by hop along the path prescribed by a monitor. The probe can follow the reverse path to return to the original monitor, thus effectively probing any path with \emph{at least} one end at a monitor.}, or at the application layer (to localize failures in overlay networks) if equivalent ``source routing'' is supported by the application. Moreover, CAP is also feasible under an emerging networking paradigm called software-defined networking (SDN) \cite{SDN,OpenFlow13}, where monitors can instruct the SDN controller to set up arbitrary paths for the probing traffic. In particular, an SDN consisting of OpenFlow switches \cite{OpenFlow13} can set up paths by configuring the flow table of each traversed OpenFlow switch to forward a probing flow (e.g., one TCP connection) to a next hop based on the ingress port and the flow identifier, which allows the path to have repeated nodes/links. In contrast, UP represents the most basic probing mechanism, feasible in any network supporting data forwarding, that provides a lower bound on the capability of failure localization. CSP represents an intermediate case that allows control over routing while respecting a basic requirement that routes must be cycle-free. CSP is implementable by MPLS (MultiProtocol Label Switching), where the ``explicit routing'' mode \cite{RFC_MPLS} allows one to set up a controllable, non-shortest path using labels so long as the path are cycle-free. Note that the cycle-free constraint here is crucial, as data forwarding in MPLS will encounter forwarding loops if a path has cycles.\looseness=-1

The significance of these three probing mechanisms is that they capture the main features of several existing and emerging routing techniques. Specifically, UP is generally supported in existing networks without special configuration, CSP is feasible in some of today's networks running MPLS with certain configuration (i.e., label propagation via explicit routing), while CAP represents the capability of future networks once SDN is broadly deployed.

\vspace{.5em}
\emph{Discussion:} In \cite{Ahuja09TON}, ``m-trail'' (monitoring trails) is employed as a probing mechanism in all-optical networks, where measurement paths can contain repeated nodes but \emph{not} repeated links. It is unclear which routing protocols in communication networks select paths under the restriction of ``m-trails'', we thus do not consider such a probing mechanism in this paper. In \cite{Cho14Elsevier}, another probing mechanism ``m-tour'' (monitoring tours) is used, which allows both repeated nodes and repeated links in measurement paths; ``m-tour'' is equivalent to CAP.
}\looseness=-1

\vspace{.5em}
{\clr In this paper, we quantify how the flexibility of a probing scheme affects the network's capability to localize failures. Although concrete results are only provided for the above classes of probing mechanisms, our framework and our abstract identifiability conditions (see Section~\ref{sec:Abstract Identifiability Conditions}) can also be used to evaluate the failure localization capabilities of other probing mechanisms.}

\subsection{Objective}

Given a network topology $\mathcal{G}$, a set of monitors $M$, and a probing mechanism (CAP, CSP, or UP), we seek to answer the following closely related questions: (i) Given a node set of interest $S$ and a bound $k$ on the number of failures, can we uniquely localize up to $k$ failed nodes in $S$ from observed path states? (ii) Given a node set $S$, what is the maximum number of failures within $S$ that can be uniquely localized? (iii) Given an integer $k$ ($1\leq k \leq \sigma$), what is the largest node set that is $k$-identifiable? We will study these problems from the perspectives of both theories and efficient algorithms.

\subsection{Illustrative Example}

\begin{figure}[tb]
\centering
\includegraphics[width=1.9in]{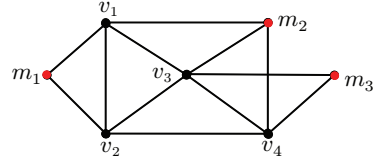}
\vspace{-.5em}
\caption{Sample network with three monitors: $m_1$, $m_2$, and $m_3$. } \label{Fig:ProblemIllustration}
\vspace{-.5em}
\end{figure}

Consider the sample network in Fig.~\ref{Fig:ProblemIllustration} with three monitors ($m_1$--$m_3$) and four non-monitors ($v_1$--$v_4$).
Under UP, suppose that the default routing protocol only allows the monitors to probe the following paths: $\mathcal{P}_1=m_1v_1m_2$, $\mathcal{P}_2=m_2v_4m_3$, and $\mathcal{P}_3=m_1v_2v_4m_3$, which form a measurement matrix $\mathbf{R}^{\mbox{\tiny UP}}$:

\scriptsize
\begin{equation}\label{eq:exampleUP}
    \left. \begin{aligned}
            \mathcal{P}_1&=m_1v_1m_2\\
            \mathcal{P}_2&=m_2v_4m_3\\
            \mathcal{P}_3&=m_1v_2v_4m_3\\
         \end{aligned} \right.
         \Rrightarrow
        \mathbf{R}^{\mbox{\tiny UP}}=\left.
        \begin{blockarray}{cccc}
            W_{1} & W_{2} & W_{3} & W_{4}\\
            \begin{block}{(cccc)}
                1 &  0 &  0 &  0\\
                0 &  0 &  0 &  1\\
                0 &  1 &  0 &  1\\
            \end{block}
        \end{blockarray}\right.,
\end{equation}
\normalsize
where ${R}^{\mbox{\tiny UP}}_{ij}=1$ if and only if node $v_j$ is on path $\mathcal{P}_i$. Then we have $\mathbf{R}^{\mbox{\tiny UP}}\odot\textbf{w}=\textbf{c}$, where $\textbf{c}$ is the binary vector of path states observed at the destination monitors. Let $S':=\{v_1,v_2,v_4\}$. Based on Definition~\ref{def:max-identifiableSet}, we can verify that $\Omega(S')=2$, and the maximum identifiable set $S^*(1)=\{v_1,v_2,v_4\}$ and $S^*(2)=S^*(3)=S^*(4)=\{v_1,v_4\}$. Under CSP, besides the three paths in (\ref{eq:exampleUP}), we can probe three additional paths: $\mathcal{P}_4=m_2v_3m_3$, $\mathcal{P}_5=m_1v_2v_3m_3$, and $\mathcal{P}_6=m_1v_2v_1m_2$, yielding an expanded measurement matrix in (\ref{eq:exampleCSP}):

\scriptsize
\begin{equation}\label{eq:exampleCSP}
    \left. \begin{aligned}
            \mathcal{P}_1&=m_1v_1m_2\\
            \mathcal{P}_2&=m_2v_4m_3\\
            \mathcal{P}_3&=m_1v_2v_4m_3\\
            \mathcal{P}_4&=m_2v_3m_3\\
            \mathcal{P}_5&=m_1v_2v_3m_3\\
            \mathcal{P}_6&=m_1v_2v_1m_2\\
         \end{aligned} \right.
         \Rrightarrow
        \mathbf{R}^{\mbox{\tiny CSP}}=\left.
        \begin{blockarray}{ccccl}
            W_{1} & W_{2} & W_{3} & W_{4}\vspace{.6em}&\\
            \begin{block}{(cccc)l}
                1 &  0 &  0 &  0 &\BAmulticolumn{1}{l}{\multirow{3}{*}{ $\left. \begin{aligned} \\ \\ \end{aligned}\right\} \begin{aligned}\mathbf{R}^{\mbox{\tiny UP}}\end{aligned}$}}\\
                0 &  0 &  0 &  1&\\
                0 &  1 &  0 &  1&\\
                \cline{1-4}
                0 &  0 &  1 &  0&\\
                0 &  1 &  1 &  0&\\
                1 &  1 &  0 &  0&\\
            \end{block}
        \end{blockarray}\right.
\end{equation}
\normalsize

Using the six paths in (\ref{eq:exampleCSP}), the maximum identifiability index of $S'$ becomes $\Omega(S')=3$, and the maximum identifiable set is enlarged to $S^*(1)=S^*(2)=S^*(3)=\{v_1,v_2,v_3,v_4\}$ and $S^*(4)=\{v_1,v_3,v_4\}$, a notable improvement over UP. Finally, if CAP is supported, then we can send probes along a cycle $\mathcal{P}_7=m_1v_2m_1$. In conjunction with the paths in (\ref{eq:exampleCSP}), this yields the measurement matrix in (\ref{eq:exampleCAP}):

\scriptsize
\begin{equation}\label{eq:exampleCAP}
    \left. \begin{aligned}
            \mathcal{P}_1&=m_1v_1m_2\\
            \mathcal{P}_7&=m_1v_2m_1\\
            \mathcal{P}_4&=m_2v_3m_3\\
            \mathcal{P}_2&=m_2v_4m_3\\
         \end{aligned} \right.
         \Rrightarrow
        \mathbf{R}^{\mbox{\tiny CAP}}=\left.
        \begin{blockarray}{cccc}
            W_{1} & W_{2} & W_{3} & W_{4}\\
            \begin{block}{(cccc)}
                1 &  0 &  0 &  0\\
                0 &  1 &  0 &  0\\
                0 &  0 &  1 &  0\\
                0 &  0 &  0 &  1\\
            \end{block}
        \end{blockarray}\right.
\end{equation}
\normalsize

Since the paths in (\ref{eq:exampleCAP}) can independently determine the state of each non-monitor, we have $\Omega(S')=4$ and $S^*(1)=S^*(2)=S^*(3)=S^*(4)=\{v_1,v_2,v_3,v_4\}$ under CAP, i.e., all failures can be uniquely localized.

This example shows that the monitor placement and the probing mechanism significantly affect a network's capability to localize failures. In the rest of the paper, we will study this relationship both theoretically and algorithmically.

\section{Theoretical Foundations}
\label{sec:TheoryFoundation}

We start with some basic understanding of failure identifiability. First, the definition of $k$-identifiability in Definition~\ref{def:identifiability} requires enumeration of all possible failure events and thus cannot be tested efficiently. To address this issue, we establish explicit sufficient/necessary conditions for $k$-identifiability that apply to arbitrary probing mechanisms, which will later be developed into verifiable conditions for the three classes of probing mechanisms. Moreover, we establish several desirable properties of maximum identifiability index (Definition~\ref{def:identifiability}) and maximum identifiable set (Definition~\ref{def:max-identifiableSet}), which greatly simplify the computation of these measures.


\subsection{Abstract Identifiability Conditions}\label{sec:Abstract Identifiability Conditions}

Our identifiability condition is inspired by a result known in a related field called \emph{combinatorial group testing} \cite{Dorfman43}. In short, group testing aims to find abnormal elements in a given set by running tests on subsets of elements, each test indicating whether any element in the subset is abnormal. This is analogous to our problem where abnormal elements are failed nodes and tests are conducted by probing measurement paths. A subtle but critical difference is that in our problem, the subsets of elements that can be tested together are constrained by the set of measurement paths $P$, which is in turn limited by the topology, probing mechanism, and placement of monitors\footnote{In this regard, our problem is similar to a variation of group testing under graph constraints \cite{Cheraghchi12}; see Section~\ref{subsec:related_work} for the difference.}.

Most existing solutions for (nonadaptive) group testing aim at constructing a \emph{disjunct testing matrix}. Specifically, a testing matrix $R$ is a binary matrix, where $R_{i,j}=1$ if and only if element $j$ is included in the $i$-th test. Matrix $R$ is \emph{$k$-disjunct} if {\clr the Boolean sum of any $k$ columns does not {\clrr ``contain''}} any other column{\clr \footnote{{\clr That is, for any subset of $k$ column indices $S$ and any other column index $j\notin S$, $\exists$a row index $i$ such that $R_{i,j} = 1$ and $R_{i,j'} = 0$ for all $j' \in S$.}}} \cite{Yeh02}. In our problem, the existence of a disjunct testing matrix translates into the following conditions.

\begin{lemma}\label{lem:abstract condition}
Set $S$ is $k$-identifiable:
\begin{enumerate}
\item[a)] if for any failure set $F$ with $|F|\leq k$ and any node $v$ with $v\in S\setminus F$, $\exists$ $p\in P$ traversing $v$ but none of the nodes in $F$;
\item[b)] only if for any failure set $F$ with $|F|\leq k-1$ and any node $v$ with $v\in S\setminus F$, $\exists$ $p\in P$ traversing $v$ but none of the nodes in $F$.
\end{enumerate}
\end{lemma}

\begin{proof}
Consider two distinct failure sets $F_1$ and $F_2$ with $F_1\cap S\neq F_2 \cap S$, each containing no more than $k$ nodes. There exists a node $v\in S$ in only one of these sets; suppose $v\in F_1\setminus F_2$. By the condition in the lemma, $\exists$ a path $p$ traversing $v$ but not $F_2$, thus distinguishing $F_1$ from $F_2$. Therefore, condition \emph{a)} in Lemma~\ref{lem:abstract condition} is sufficient.

Suppose $\exists$ a non-empty set $F$ with $|F|\leq k-1$ and $v\in S\setminus F$ such that all measurement paths traversing $v$ must also traverse at least one node in $F$. Therefore, for two failure sets $F$ and $F\cup \{v\}$ satisfying conditions (1--2) in Definition~\ref{def:identifiability}-(1) are not distinguishable as $P_{F}=P_{F\cup \{v\}}$. Thus, condition \emph{b)} in Lemma~\ref{lem:abstract condition} is necessary.
\end{proof}

These conditions generally apply to any probing mechanism. Although in the current form, they do not directly lead to efficient testing algorithms, we will show later (Section~\ref{sec:Verifiable Identifiability Conditions}) that they can be transformed into verifiable conditions for several classes of probing mechanisms.

\subsection{Properties of the Maximum Identifiability Index and the Maximum Identifiable Set}
\label{sec:setProperty}
Although the maximum identifiability index $\Omega(S)$ and the maximum $k$-identifiable set $S^*(k)$ are defined for sets of nodes, we show below that they can both be characterized in terms of a per-node property, which greatly simplifies the computation of these measures. We start with the following two observations.

\begin{lemma}\label{lem:nodes_in_OmegaS}
\begin{enumerate}
\item[a)] If $S$ is $k$-identifiable, then any $v\in S$ must be $k$-identifiable.
\item[b)] If $v$ is $k$-identifiable $\forall v\in S$, then $S$ is $k$-identifiable.
\end{enumerate}
\end{lemma}
\begin{proof}
a) Suppose $\exists$ node $v\in S$ that is not $k$-identifiable, then $\exists$ at least two failure sets $F_1$ and $F_2$ with $|F_i|\leq k$ ($i=\{1,2\}$) and $F_1 \cap \{v\} \neq F_2 \cap \{v\}$ such that $F_1$ and $F_2$ are not distinguishable. Thus, $S$ is not $k$-identifiable as $v\in S$.\looseness=-1

b) For any two failure sets $F_1$ and $F_2$ with $|F_i|\leq k$ ($i=\{1,2\}$) and $F_1 \cap S \neq F_2 \cap S$, $\exists$ a node $v\in S$ that is either in $F_1$ or $F_2$ but not both. Since node $v$ is $k$-identifiable, $F_1$ and $F_2$ must be distinguishable. Therefore, $S$ is $k$-identifiable.
\end{proof}

\begin{proposition}
\label{prop:Construct_Omega_S}
$\Omega(S) = \min_{v\in S}\Omega(v)$.
\end{proposition}
\begin{proof}
By Lemma~\ref{lem:nodes_in_OmegaS}-(a), any $v\in S$ must have $\Omega(v)\geq \Omega(S)$. Thus, $\min_{v\in S}\Omega(v) \geq \Omega(S)$. By the definition of maximum identifiability index, all nodes in $S$ are $\min_{v\in S}\Omega(v)$-identifiable. By Lemma~\ref{lem:nodes_in_OmegaS}-(b), $S$ is also $\min_{v\in S}\Omega(v)$-identifiable. Thus, $\Omega(S) \geq \min_{v\in S}\Omega(v)$. Therefore, $\Omega(S) = \min_{v\in S}\Omega(v)$.
\end{proof}

\begin{corollary}
\label{coro:Omega_decreasing}
Maximum identifiability index of $S$, $\Omega(S)$, is monotonically non-increasing in the sense that $\Omega(S_1)\geq \Omega(S_2)$ for any two non-empty sets $S_1$ and $S_2$ with $S_1\subset S_2$.
\end{corollary}
\begin{proof}
Since $S_1\subset S_2$, $\min_{v\in S_1}\Omega(v) \geq \min_{v\in S_2}\Omega(v)$. Therefore, by Proposition~\ref{prop:Construct_Omega_S}, $\Omega(S_1)\geq \Omega(S_2)$.
\end{proof}

Therefore, we can estimate the maximum identifiability index of a given non-monitor set using Corollary~\ref{coro:Omega_decreasing} when the maximum identifiability index of its subset/superset is known.

Next, we show that maximum $k$-identifiable sets exhibit properties that can facilitate fast determination of which nodes should be included/excluded in these sets.

\begin{proposition}
\label{prop:construct_maximum_identifiable_set}
Let $S'(k):=\{v\in N:$ $v$ is $k$-identifiable $\}$. Then $S'(k)=S^*(k)$.
\end{proposition}
\begin{proof}
By Lemma~\ref{lem:nodes_in_OmegaS}-(a), any node in $S^*(k)$ is $k$-identifiable. Therefore, $S^*(k)\subseteq S'(k)$.

Next, $S'(k)$ must be $k$-identifiable according to Lemma~\ref{lem:nodes_in_OmegaS}-(b). Thus $|S'(k)|\leq |S^*(k)|$. Consequently, $S'(k)=S^*(k)$.
\end{proof}

Proposition~\ref{prop:construct_maximum_identifiable_set} provides a method to construct the maximum $k$-identifiable set $S^*(k)$ by simply collecting all $k$-identifiable nodes. Based on this method, we can further prove the uniqueness and monotonicity of $S^*(k)$ as follows:

\begin{corollary}
\label{coro:unique_maximum_identifiable_set}
The maximum $k$-identifiable set $S^*(k)$ is unique and monotonically non-increasing in $k$, i.e., $S^*({k+1})\subseteq S^*(k)$ for any $k$.
\end{corollary}
\begin{proof}
Definition~\ref{def:identifiability} implies that $k$-identifiability is a per-node property that is independent of the identifiability of other nodes. Therefore, for each node in $N$, it is either $k$-identifiable or \emph{not} $k$-identifiable. By Proposition~\ref{prop:construct_maximum_identifiable_set}, $S^*(k)$ is a set containing all $k$-identifiable nodes; therefore, $S^*(k)$ is unique.\looseness=-1

For each node $w\in N\setminus S^*(k)$, $w$ is \emph{not} $k$-identifiable, and thus $w$ is not $(k+1)$-identifiable. Since $S^*(k+1)$ is a collection of all $(k+1)$-identifiable nodes, no nodes in $N\setminus S^*(k)$ can be included in $S^*(k+1)$. Thus, $S^*({k+1})\subseteq S^*(k)$.
\end{proof}

Intuitively, if there exists a $k$-identifiable set $S'(k)$ with $|S'(k)|=|S^*(k)|$, then we must have $S'(k)=S^*(k)$. Thus, Corollary~\ref{coro:unique_maximum_identifiable_set} suggests one way to obtain $S^*(k)$ is to identify $S^*(j)$ for $j<k$ and then only study subsets of $S^*(j)$; nodes outside $S^*(j)$ are guaranteed to be excluded from $S^*(k)$.\looseness=-1

\begin{corollary}
\label{coro:maximum_set_lower_bound}
Let $S''(k):=\{v\in N:\exists$ path in $P$ traversing $v$ but none of the nodes in each failure set $F$ with $v\notin F$ and $|F|\leq k\}$. Then $S''(k)\subseteq S^*(k)$.
\end{corollary}
\begin{proof}
$S''(k)$ satisfies sufficient condition a) in Lemma~\ref{lem:abstract condition}. Thus, $\Omega\big(S''(k)\big)\geq k$. Following similar arguments as in the proof of Proposition~\ref{prop:construct_maximum_identifiable_set}, again we have that each node in $S''(k)$ is at least $k$-identifiable. Therefore, $S''(k)\subseteq S^*(k)$.
\end{proof}

By Corollary~\ref{coro:maximum_set_lower_bound}, we note that $S''(k)$ underestimates the size of the maximum $k$-identifiable set $S^*(k)$, yet it forms an inner bound (i.e., subset) of $S^*(k)$, thus providing theoretical support for determining the must-have nodes in the optimum set $S^*(k)$; see detailed discussions presented in Section~\ref{sec:CharacterizationMaximumIdentifiableSet}.


\vspace{.5em}
\emph{Remark:} Results in this section apply to any probing mechanism. We will show in the following sections how they can be used to design efficient algorithms for probing mechanisms CAP, CSP, and UP. The above results can also be used to design algorithms for other probing mechanisms.

\section{Verifiable Identifiability Conditions}
\label{sec:Verifiable Identifiability Conditions}

In this section, starting from the abstract conditions in Section~\ref{sec:Abstract Identifiability Conditions}, we develop concrete conditions suitable for efficient testing for the three classes of probing mechanisms.

\subsection{Conditions under CAP}\label{subsec:Conditions for k-identifiability, CAP}

CAP essentially allows us to ``ping'' any node from a monitor along any path. In the face of failures, this allows a monitor to determine the state of a node as long as it is connected to the node after removing other failed nodes. This observation allows us to translate the conditions in Section~\ref{sec:Abstract Identifiability Conditions} into more concrete identifiability conditions (Lemma~\ref{lem:k-identifiability, CAP}).

\begin{lemma}\label{lem:k-identifiability, CAP}
Set $S$ is $k$-identifiable under CAP if and only if for any set $V'$ of up to $k-1$ non-monitors, each connected component in $\mathcal{G}-V'$ that contains a node in $S$ has a monitor.
\end{lemma}

\begin{proof}
\emph{Necessity.} Suppose the above condition does not hold, i.e., there exists a non-monitor $v$ ($v\in S$) that is disconnected from all monitors in $\mathcal{G}-V'$ for a set $V'$ of up to $k-1$ non-monitors ($v\not\in V'$). Then if nodes in $V'$ fail, no remaining measurement path can probe $v$, violating the condition in Lemma~\ref{lem:abstract condition}-(b).

\emph{Sufficiency.} The proof is similar to that of Theorem~2 in \cite{Cho14Elsevier}, except that we are only interested in localizing failures in $S$. Consider two failure sets $F_1$ and $F_2$ with $|F_i|\leq k$ ($i=\{1,2\}$) and $F_1 \cap S \neq F_2 \cap S$. Then $\exists$ node $v$ ($v\in S$) that is in one and only one of $F_1$ and $F_2$. Without loss of generality, let $v\in F_1$. Let $I:=F_1 \cap F_2$. Since $|I|\leq k-1$, $\exists$ a path $p$ connecting a monitor $m$ with node $v$ in $\mathcal{G}-I$ if the condition in Lemma~\ref{lem:k-identifiability, CAP} holds. Let $w$ be the first node on $p$ (starting from $m$) that is in either $F_1\setminus I$ or $F_2\setminus I$. Truncating $p$ at $w$ gives a path $p'$ such that $p'$ and its reverse path form a measurement path from $m$ to $w$ and back to $m$ that traverses only $F_1$ or $F_2$, thus distinguishing $F_1$ and $F_2$.
\end{proof}

Under CAP, Lemma~\ref{lem:k-identifiability, CAP} shows that the necessary condition derived from Lemma~\ref{lem:abstract condition} is also sufficient. However, the condition in Lemma~\ref{lem:k-identifiability, CAP} still cannot be tested efficiently because a combinatorial number of sets $V'$ are enumerated. Fortunately, we can reduce Lemma~\ref{lem:k-identifiability, CAP} into explicit conditions on vertex-cuts of a related topology, which can then be tested in polynomial time. We use the following notion from graph theory.\looseness=-1

\begin{definition}\label{def:vertex-cut}
For two nodes $s$ and $t$ in an undirected graph $\mathcal{G}$, $(s,t)$-vertex-cut in $\mathcal{G}$, denoted by $C_{\mathcal{G}}(s,t)$, is the minimum-cardinality node set whose deletion destroys all paths from $s$ to $t$. If $s$ and $t$ are neighbors, $C_{\mathcal{G}}(s,t):=V(\mathcal{G})\setminus \{t\}$.
\end{definition}

\begin{figure}[tb]
\centering
\includegraphics[width=3.4in]{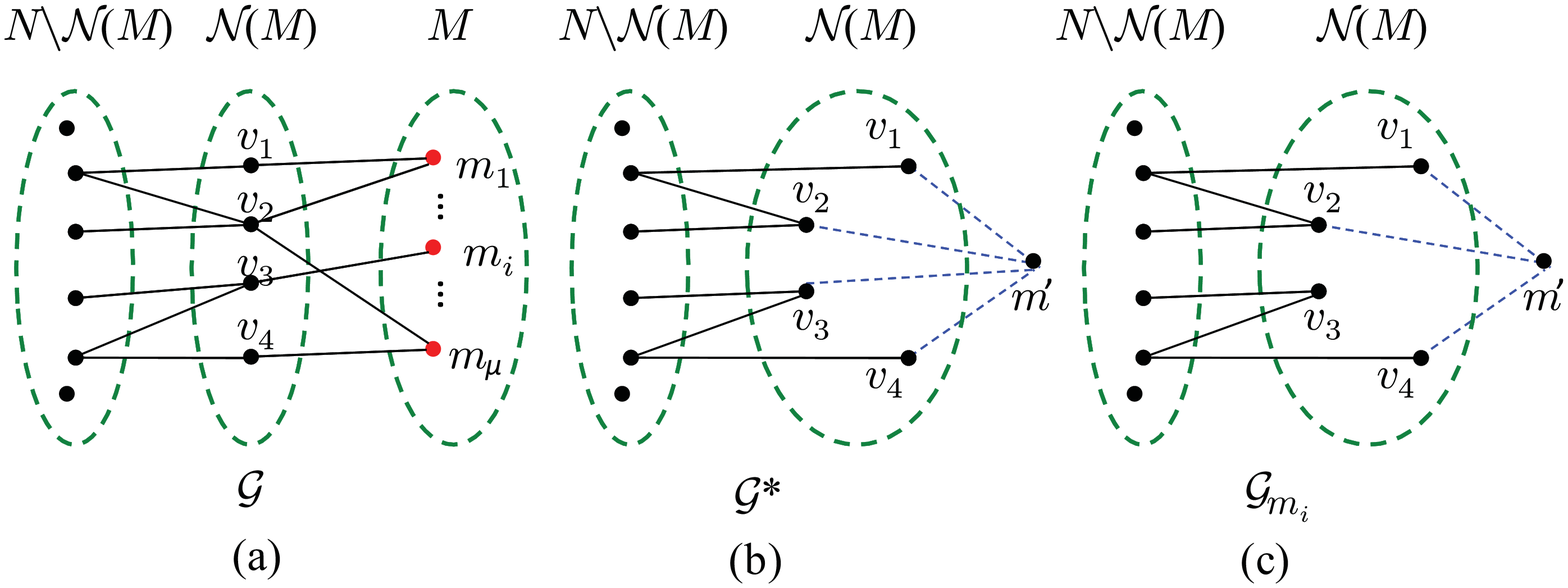}
\vspace{-0.5em}
\caption{Auxiliary graphs: (a) Original graph $\mathcal{G}$; (b) $\mathcal{G}^*$ of $\mathcal{G}$; (c) $\mathcal{G}_{m_i}$ of $\mathcal{G}$ w.r.t. monitor $m_i$.} \label{Fig:AuxiliaryGraph}
\vspace{-1.5em}
\end{figure}

Our key observation is that requiring each connected component in $\mathcal{G}-V'$ that contains a node in $S$ to have a monitor is equivalent to requiring each such component in $\mathcal{G}-M-V'$ (i.e., after removing all monitors) to contain a neighbor of a monitor. Thus, if we extend $\mathcal{G}-M$ by adding a \emph{virtual monitor} $m'$ and \emph{virtual links} connecting $m'$ and all neighbors of monitors to obtain an \emph{auxiliary graph} $\mathcal{G}^*:= \mathcal{G}-M +\{m'\} + \mathcal{L}\big(\{m'\},\: \mathcal{N}(M)\big)$ (illustrated in Fig.~\ref{Fig:AuxiliaryGraph}~(b)), then each node in $S\setminus V'$ should be connected to $m'$ in $\mathcal{G}^*-V'$. In other words, the minimum cardinality of the ($m',w$)-vertex-cut in $\mathcal{G}^*$ over all $w\in S$ must be greater than $|V'|$. For ease of presentation, we introduce the following definition.\looseness=-1

\begin{definition}
\label{def:CardinalityOfVertexCut}
Given a graph $\mathcal{G}$, a node set $S$, and a node $m\notin S$, define $\Gamma_{\mathcal{G}}(S, m) := \min_{w\in S}|C_{\mathcal{G}}(w,m)|$.
\end{definition}

By this definition, Lemma~\ref{lem:k-identifiability, CAP} can be transformed into a new condition, which reduces the tests over all possible $V'$ to a single test of the vertex-cuts of $\mathcal{G}^*$, as stated below (recall that $\sigma$ is the total number of non-monitors).

\begin{lemma}\label{lem:connectivity of G*}
Each connected component in $\mathcal{G}-V'$ that contains a node in $S$ has a monitor for any set $V'$ of up to $q$ ($q\leq \sigma-1$) non-monitors if and only if $\Gamma_{\mathcal{G}^*}(S, m')\geq q+1$.
\end{lemma}

\begin{proof}
The proof can be found in \cite{MaTONTR_Mar2015}.
\end{proof}

Lemma~\ref{lem:connectivity of G*} allows us to rewrite the identifiability conditions in Lemma~\ref{lem:k-identifiability, CAP} in terms of the vertex-cuts of $\mathcal{G}^*$.

\begin{theorem}[$k$-identifiability under CAP]\label{thm:k-identifiability, CAP}
Set $S$ is $k$-identifiable ($k\leq \sigma$) under CAP if and only if $\Gamma_{\mathcal{G}^*}(S, m')\geq k$.
\end{theorem}

A special case of Theorem~\ref{thm:k-identifiability, CAP} occurs when $k=\sigma$, i.e., any non-monitors can fail simultaneously. In this case, each node in $S$ must directly connect to at least one monitor in $\mathcal{G}$.

\emph{Discussion:} Theorem~\ref{thm:k-identifiability, CAP} extends and improves the identifiability condition given in Theorem~2 of \cite{Cho14Elsevier} by (i) considering failures within an arbitrary subset of nodes instead of the entire network, and (ii) providing a single condition that can be tested in polynomial time (see testing algorithm below) instead of testing a combinatorial number of conditions that enumerate all possible failure events.

\vspace{.5em}
\textbf{Testing algorithm:} A key advantage of the newly derived conditions over the abstract conditions in Section~\ref{sec:Abstract Identifiability Conditions} is that they can be tested efficiently. Let $\theta:=|\mathcal{N}(M)|$ denote the number of non-monitors that are neighbors of at least one monitor in $M$. Given node $w$, $C_{\mathcal{G}^*}(w,m')$ can be computed in $O(\theta \xi)$ time\footnote{The ($m',w$)-vertex-cut problem in an undirected graph can be reduced to an ($m',w$)-edge-cut problem in a directed graph in linear time \cite{Chuzhoy09JACM}. The ($m',w$)-edge-cut problem is solvable by the Ford$-$Fulkerson algorithm \cite{Ford56} in $O(\theta \xi)$ time.
}, where $\xi$ is the number of links (refer to Table~\ref{t notion} for notations). Therefore, we can evaluate $\Gamma_{\mathcal{G}^*}(S, m')$ in $O(\theta \xi |S|)$ time and compare the result with $k$ to test the conditions in Theorem~\ref{thm:k-identifiability, CAP}.

\subsection{Conditions under CSP}\label{subsec:Conditions for k-identifiability, CSP}

Under CSP, we restrict measurement paths $P$ to the set of \emph{simple paths} between monitors, i.e., paths starting/ending at distinct monitors that contain no cycles. As in CAP, our goal is again to transform the abstract conditions in Section~\ref{sec:Abstract Identifiability Conditions} into concrete sufficient/necessary conditions that can be efficiently verified. We first give analogous result to Theorem~\ref{thm:k-identifiability, CAP}.

\begin{lemma}\label{lem:k-identifiability, CSP}
Set $S$ is $k$-identifiable under CSP:
\begin{enumerate}
\item[a)] if for any node set $V'$, $|V'|\leq k+1$, containing at most one monitor, each connected component in $\mathcal{G}-V'$ that contains a node in $S$ also contains a monitor;
\item[b)] only if for any node set $V'$, $|V'|\leq k$, containing at most one monitor, each connected component in $\mathcal{G}-V'$ that contains a node in $S$ also contains a monitor.
\end{enumerate}
\end{lemma}

\begin{proof}
The proof can be found in \cite{MaTONTR_Mar2015}.
\end{proof}

Due to the restriction to simple paths, the identifiability conditions in Lemma~\ref{lem:k-identifiability, CSP} are stronger than those in Lemma~\ref{lem:k-identifiability, CAP}. As with Lemma~\ref{lem:k-identifiability, CAP}, the conditions in Lemma~\ref{lem:k-identifiability, CSP} do not directly lead to efficient tests, and we again seek equivalent conditions in terms of topological properties. Each condition in the form of Lemma~\ref{lem:k-identifiability, CSP} (a--b) covers two cases: (i) $V'$ only contains non-monitors; (ii) $V'$ contains a monitor and $|V'|-1$ non-monitors. The first case has been converted to a vertex-cut property on an auxiliary topology $\mathcal{G}^*$ by Lemma~\ref{lem:connectivity of G*}; we now establish a similar condition for the second case using similar arguments.\looseness=-1

Fix a set $V'=F \cup \{m\}$, where $m$ is a monitor in $M$ and $F$ a set of non-monitors. Again, the key observation is that each connected component in $\mathcal{G}-V'$ that contains a node in $S$ also containing a monitor is equivalent to each such component in $\mathcal{G}-M-F$ containing a neighbor of a monitor other than $m$ (i.e., a node in $\mathcal{N}(M\setminus \{m\})$). To capture this observation, we introduce another \emph{auxiliary graph} $\mathcal{G}_m := \mathcal{G} - M + \{m'\} + \mathcal{L}\big(\{m'\},\: \mathcal{N}(M\setminus \{m\})\big)$ w.r.t. monitor $m$ as illustrated in Fig.~\ref{Fig:AuxiliaryGraph}~(c), where $m'$ is again a virtual monitor. We will show that the second case ($V'$ contains a monitor) is equivalent to requiring that the nodes in $S\setminus F$ and $m'$ are in the same connected component within $\mathcal{G}_m - F$, and thus the following holds.

\begin{lemma}\label{lem:connectivity of G_i}
The following two conditions are equivalent:
\begin{enumerate}
  \item[(1)] Each connected component in $\mathcal{G}-V'$ that contains a node in $S$ also contains a monitor for $\forall$sets $V'$ containing monitor $m$ ($m\in M$) and up to $q$ ($q\leq \sigma-1$) non-monitors;
  \item[(2)] $\Gamma_{\mathcal{G}_m}(S,m')\geq q+1$.
\end{enumerate}
\end{lemma}

\begin{proof}
The proof can be found in \cite{MaTONTR_Mar2015}.
\end{proof}

Based on Lemmas~\ref{lem:connectivity of G*} and \ref{lem:connectivity of G_i}, we can rewrite Lemma~\ref{lem:k-identifiability, CSP} as follows.

\begin{theorem}[$k$-identifiability under CSP]\label{thm:k-identifiability, CSP}
Set $S$ is $k$-identifiable under CSP:
\begin{enumerate}
\item[a)] if $\Gamma_{\mathcal{G}^*}(S,m')\geq k+2$, and $\min_{m\in M}\Gamma_{\mathcal{G}_m}(S,m')\geq k+1$ ($k\leq \sigma-2$);
\item[b)] only if $\Gamma_{\mathcal{G}^*}(S,m')\geq k+1$, and $\min_{m\in M}\Gamma_{\mathcal{G}_m}(S,m')\geq k$ ($k\leq \sigma-1$).
\end{enumerate}
\end{theorem}

Theorem~\ref{thm:k-identifiability, CSP} does not include the cases of $k=\sigma$ and $k=\sigma-1$, which are addressed in Propositions~\ref{prop:sigma-identifiability, CSP} and \ref{prop:(sigma-1)-identifiability, CSP}.\looseness=-1

\begin{proposition}\label{prop:sigma-identifiability, CSP}
Set $S$ is $\sigma$-identifiable under CSP if and only if each node in $S$ has at least two monitors as neighbors.\looseness=-1
\end{proposition}

\begin{proof}
The proof can be found in \cite{MaTONTR_Mar2015}.
\end{proof}

\begin{proposition}\label{prop:(sigma-1)-identifiability, CSP}
Set $S$ is $(\sigma-1)$-identifiable under CSP if and only if (i) all nodes in $S$ have at least two monitors as neighbors, or (ii) all nodes in $N\setminus \{v\}$ ($v\in S$) have at least two monitors as neighbors and $v$ has all nodes in $N\setminus \{v\}$ and one monitor as neighbors.
\end{proposition}

\begin{proof}
The proof can be found in \cite{MaTONTR_Mar2015}.
\end{proof}

\textbf{Testing algorithm:} Similar to the case of CAP, we can use the algorithm in \cite{Chuzhoy09JACM,Ford56} to compute the vertex-cuts of the auxiliary graphs $\mathcal{G}^*$ and $\mathcal{G}_m$ ($\forall m\in M$), and test the conditions in Theorem~\ref{thm:k-identifiability, CSP} for any given $k$. The overall complexity of the test is $O(\mu \theta \xi |S|)$ (refer to Table~\ref{t notion} for notations).

\subsection{Conditions under UP}\label{subsec:Conditions for k-identifiability, UP}

Under UP, monitors have no control over the probing paths between monitors, and the set of measurement paths $P$ is limited to the paths between monitors determined by the network's native routing protocol. In contrast to the previous cases (CAP, CSP), identifiability under UP can no longer be characterized in terms of topological properties. We can, nevertheless, establish explicit conditions based on the abstract conditions in Section~\ref{sec:Abstract Identifiability Conditions}. The idea is to examine how many non-monitors need to be removed to disconnect all measurement paths traversing a given non-monitor $v$. If the number is sufficiently large (greater than $k$), then we can still infer the state of $v$ from some measurement path when a set of other non-monitors fail; if the number is too small (smaller than or equal to $k-1$), then we are not able to determine the state of $v$ as the failures of all paths traversing $v$ can already be explained by the failures of other non-monitors. This intuition leads to the following results.

In the sequel, $P_v\subseteq P$ denotes the set of measurement paths traversing a non-monitor $v$, and $\mathcal{C}_v:= \{P_w:\: w\in N,\: w\neq v\}$ denotes the collection of path sets traversing non-monitors in $N\setminus \{v\}$. We use $\MSC(v)$ to denote the size of the \emph{minimum set cover} of $P_v$ by $\mathcal{C}_v$, i.e., $\MSC(v):= |V'|$ for the minimum set $V'\subseteq N\setminus \{v\}$ such that $P_v\subseteq \bigcup_{w\in V'}P_w$. Note that covering is only feasible if $v$ is not on any $2$-hop measurement path (i.e., monitor-$v$-monitor), in which case we know $P_v\subseteq \bigcup_{w\in N, w\neq v}P_w$ and thus $\MSC(v)\leq \sigma-1$. If $v$ is on a $2$-hop path, then we define $\MSC(v):= \sigma$.

\begin{theorem}[$k$-identifiability under UP]\label{thm:k-identifiability, UP}
Set $S$ is $k$-identifiable under UP with measurement paths $P$: \begin{enumerate}
\item[a)] if $\MSC(v)\geq k+1$ for any node $v$ in $S$ ($k\leq \sigma-1$);
\item[b)] only if $\MSC(v)\geq k$ for any node $v$ in $S$ ($k\leq \sigma$).
\end{enumerate}
\end{theorem}

\begin{proof}
The proof can be found in \cite{MaTONTR_Mar2015}.
\end{proof}

The only case not considered by Theorem~\ref{thm:k-identifiability, UP} is the case that $k=\sigma$, for which we develop the following condition.

\begin{proposition}
\label{prop:sigma-identifiability_UP}
Set $S$ is $\sigma$-identifiable under UP if and only if $MSC(v)=\sigma$ for any node $v$ in $S$, i.e., each node in $S$ is on a $2$-hop path.
\end{proposition}
\begin{proof}
The proof can be found in \cite{MaTONTR_Mar2015}.
\end{proof}

\textbf{Testing algorithm:}
The conditions in Theorem~\ref{thm:k-identifiability, UP} provide an explicit way to test $k$-identifiability under UP, using tests of the form $\MSC(v)\geq q$. Unfortunately, evaluating such a test, known as the decision problem of the \emph{set covering problem}, is known to be NP-complete. Nevertheless, we can use approximation algorithms to compute bounds on $\MSC(v)$. An algorithm with the best approximation guarantee is the \emph{greedy algorithm}, which iteratively selects the set in $\mathcal{C}_v$ that contains the largest number of uncovered paths in $P_v$ until all the paths in $P_v$ are covered (assuming that $v$ is not on any $2$-hop path).\looseness=-1

Let $\GSC(v)$ denote the number of sets selected by the greedy algorithm. This immediately provides an upper bound: $\MSC(v)\leq \GSC(v)$. Moreover, since the greedy algorithm has an approximation ratio of $\log(|P_v|)+1$ \cite{Chvatal79}, we can also bound $\MSC(v)$ from below: $\MSC(v) \geq \GSC(v)/(\log(|P_v|)+1)$. Applying these bounds to Theorem~\ref{thm:k-identifiability, UP} yields relaxed conditions:\looseness=-1
\begin{itemize}
\item $S$ is $k$-identifiable under UP if  $k < \lceil \min_{v\in S}{\frac{\GSC(v)}{\log(|P_v|)+1}} \rceil$;
\item $S$ is \emph{not} $k$-identifiable under UP if $k > \min_{v\in S}\GSC(v)$.
\end{itemize}
These conditions can be tested by running the greedy algorithm for all nodes in $S$, each taking time $O(|P_v|^2\sigma)=O(|P|^2\sigma)$, and the overall test has a complexity of $O(|S||P|^2\sigma)$ (or $O(\mu^4\sigma |S|)$ if there is a measurement path between each pair of monitors).

\subsection{{\clrr Special Case: 1-identifiability}}\label{subsec:test of 1-identifiability}

In practice, the most common failure event consists of the failure of a single node. Thus, an interesting question is whether $S$ is $1$-identifiable under a given monitor placement and a given probing mechanism. In our previous results, Theorems~\ref{thm:k-identifiability, CSP} and \ref{thm:k-identifiability, UP} only provide an answer to the above question if the sufficient condition is satisfied or the necessary condition is violated for $k=1$; however, the answer is unknown if $S$ satisfies the necessary condition but violates the sufficient condition under CSP and UP. In contrast, Theorem~\ref{thm:k-identifiability, CAP} establishes a condition under CAP that is both necessary and sufficient, yet still expressed in a complicated form (i.e., vertex-cuts). We develop explicit methods below for testing $S$ for $1$-identifiability.

\subsubsection{Conditions for $1$-identifiability}\label{subsubsec:Abstract Condition for 1-identifiability}

We start with a generic necessary and sufficient condition that applies to all probing mechanisms. Recall that $P_v$ denotes the set of measurement paths traversing a non-monitor $v$. For $k=1$, Definition~\ref{def:identifiability}-(1) is equivalent to the following:

\begin{claim}\label{claim:1identifiability}
$S$ is $1$-identifiable if and only if:
\begin{enumerate}
\item[(1)] $P_v\neq \emptyset$ for any $v\in S$, and
\item[(2)] $P_v\neq P_w$ for any $v\in S$, $w\in N$, and $v\neq w$.
\end{enumerate}
\end{claim}

In Claim~\ref{claim:1identifiability}, the first condition guarantees that any failure in $S$ is detectable (i.e., causing at least one path failure), and the second condition guarantees that the observed path states can uniquely localize the failed node in $S$.
An efficient test of these conditions, however, requires different strategies for different probing mechanisms.

\subsubsection{Test under CAP}

By Theorem~\ref{thm:k-identifiability, CAP}, $S$ is $1$-identifiable under CAP if and only if $\Gamma_{\mathcal{G}^*}(S,m')\geq 1$. This is equivalent to requiring that $\mathcal{G}^*$ be connected, i.e., $\mathcal{G}$ has one monitor.

Testing for $1$-identifiability of $S$ under CAP is therefore reduced to determining if the network has a monitor.

\subsubsection{Test under CSP}

Under CSP, we derive conditions that are equivalent to those in Claim~\ref{claim:1identifiability} but easier to test.

\begin{figure}[tb]
\centerline{
\includegraphics[width=.28\columnwidth]{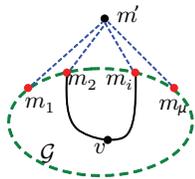}
}
\vspace{-.5em}
\caption{\footnotesize Extended graph $\mathcal{G}'$. } \label{fig:extended_graph}
\vspace{-1.5em}
\end{figure}

Condition (1) in Claim~\ref{claim:1identifiability} requires that every non-monitor in $S$ reside on a monitor-monitor simple path. While an exhaustive search for such a path incurs an exponential cost, we can test for its existence efficiently using the following observation. The idea is to construct an \emph{extended graph} $\mathcal{G}' := \mathcal{G} + \{m'\} + \mathcal{L}(\{m'\},\: M)$, i.e., by adding a virtual monitor $m'$ and connecting it to all the monitors; see an illustration in Fig.~\ref{fig:extended_graph}. We claim that a non-monitor $v$ is on a monitor-monitor simple path if and only if the size of the ($m',v$)-vertex-cut in $\mathcal{G}'$ is at least two, i.e., $\Gamma_{\mathcal{G}'}(v,m')\geq 2$, which implies the existence (see Definition~\ref{def:vertex-cut}) of two vertex-independent simple paths between $v$ and $m'$, illustrated as paths $vm_2m'$ and $vm_im'$ in Fig.~\ref{fig:extended_graph}. Truncating these two paths at $m_2$ and $m_i$ yields two path segments $vm_2$ and $vm_i$, whose concatenation gives a monitor-to-monitor simple path traversing $v$, i.e., $m_2vm_i$ in Fig.~\ref{fig:extended_graph}. On the other hand, if $\exists$ a monitor-to-monitor simple path traversing $v$, then it can be split into two simple paths connecting $v$ to two distinct monitors, which implies $\Gamma_{\mathcal{G}'}(v,m')\geq 2$ as each of these two distinct monitors connects to $m'$ by a virtual link.

Condition (2) in Claim~\ref{claim:1identifiability} is violated if and only if there exist two non-monitors $v\neq w$ (at least one of them in $S$) such that all monitor-to-monitor simple paths traversing $v$ must traverse $w$ (i.e., $P_v\subseteq P_w$) and vice versa. Since $P_v\subseteq P_w$ means that there is no monitor-to-monitor simple path traversing $v$ in $\mathcal{G}-\{w\}$, by the above argument, we see that $P_v\subseteq P_w$ if and only if the size of the ($m',v$)-vertex-cut in a new graph $\mathcal{G}'_w := \mathcal{G}-\{w\}+\{m'\}+\mathcal{L}(\{m'\},\: M)$ is smaller than two. Therefore, condition (2) in Claim~\ref{claim:1identifiability} is satisfied if and only if for every two distinct non-monitors $v$ ($v\in S$) and $w$, either the ($m',v$)-vertex-cut in $\mathcal{G}'_w$ or the ($m',w$)-vertex-cut in $\mathcal{G}'_v$ contains two or more nodes.

In summary, the necessary and sufficient condition for $1$-identifiability under CSP is:
\begin{enumerate}
\item[i)] $\Gamma_{\mathcal{G}'}(S,m')\geq 2$, and
\item[ii)] $\Gamma_{\mathcal{G}_w'}(v,m')\geq 2$ or $\Gamma_{\mathcal{G}_v'}(w,m')\geq 2$ for all $v\in S$, $w\in N$, and $v\neq w$.
\end{enumerate}
Since $\Gamma_{\mathcal{G}}(v,w) \geq 2$ can be tested in $O(|V|+|L|)$ time\footnote{We can compute the biconnected component decomposition \cite{Tarjan72} and test if $v$ and $w$ belong to the same biconnected component.}, the overall test takes $O(\sigma|S| (|V|+|L|))=O(\sigma(\mu+\sigma)^2|S|)$ time.

\subsubsection{Test under UP}

Under UP, the total number of measurement paths $|P|$ is reduced to $O(\mu^2)$ (from exponentially many as in the case of CAP/CSP) as the measurable routes are predetermined. This reduction makes it feasible to directly test conditions (1--2) in Claim~\ref{claim:1identifiability} by testing condition (1) for each node in $S$ and condition (2) for each pair of non-monitors (one of which is in $S$). Then the overall complexity of is $O(\sigma\mu^2|S|)$, dominated by testing of condition (2) in Claim~\ref{claim:1identifiability}.

\section{Characterization of Maximum Identifiability Index}\label{sec:Characterization of Maximum Identifiability}

By Proposition~\ref{prop:Construct_Omega_S}, the maximum identifiability index of a given set $S$ is the minimum per-node maximum identifiability index $\Omega(v)$ for each node $v \in S$. It thus suffices to characterize the per-node maximum identifiability index for each probing mechanism. Under CAP, we give the exact value of $\Omega(v)$ based on the necessary and sufficient condition in Theorem~\ref{thm:k-identifiability, CAP}; under CSP and UP, we establish tight upper and lower bounds on $\Omega(v)$ based on the conditions in Theorems~\ref{thm:k-identifiability, CSP} and \ref{thm:k-identifiability, UP}.\looseness=-1

\subsection{Maximum Identifiability Index under CAP}

Since Theorem~\ref{thm:k-identifiability, CAP} provides necessary and sufficient conditions, it directly determines the value of $\Omega(v)$, as stated below.\looseness=-1

\begin{theorem}[Maximum Per-node Identifiability under CAP]\label{thm:Omega_CAP}
The maximum identifiability index of a non-monitor $v$ under CAP is $\OmegaCAP(v)=\Gamma_{\mathcal{G}^*}(v,m')$.
\end{theorem}


\textbf{Evaluation algorithm:}
As shown in Section~\ref{subsec:Conditions for k-identifiability, CAP}, $\Gamma_{\mathcal{G}^*}(v,m')$ can be computed in $O(\theta \xi)$ time ($\theta$: the number of monitor neighbors in $\mathcal{G}$, $\xi$: the number of links in $\mathcal{G}$; see Table~\ref{t notion}). Therefore, $\OmegaCAP(S)$ is computable in $O(\theta\xi|S|)$ time.

\subsection{Maximum Identifiability Index under CSP}

Observing that both the sufficient and the necessary conditions in Theorem~\ref{thm:k-identifiability, CSP} are imposed on the same property, i.e., vertex-cuts of the auxiliary graph $\mathcal{G}^*$ and $\mathcal{G}_m$. Let $\delta^*:=\Gamma_{\mathcal{G}^*}(v,m')$, $\delta_{\min}:=\min_{m\in M} \Gamma_{\mathcal{G}_m} (v,m')$, and $\pi_v:=\min(\delta_{\min},\delta^*-1)$. We obtain a tight characterization of the maximum identifiability index under CSP as follows.

\begin{theorem}[Maximum Per-node Identifiability under CSP]\label{thm:Omega_CSP}
If $\pi_v \leq \sigma-2$, the maximum identifiability index of a non-monitor $v$ under CSP is bounded by $\pi_v - 1 \leq \OmegaCSP(v) \leq \pi_v$.
\end{theorem}

\begin{proof}
The proof can be found in \cite{MaTONTR_Mar2015}.
\end{proof}

\emph{Remark:} Because the set of links in $\mathcal{G}_m$ is a subset of those in $\mathcal{G}^*$ while the nodes are the same, we always have $\deltaMin \leq \delta^*$. Therefore, the above bounds simplify to: \begin{itemize}
\item $\deltaMin-2 \leq \OmegaCSP(v) \leq \deltaMin-1$ if $\deltaMin = \delta^*$;
\item $\deltaMin-1 \leq \OmegaCSP(v) \leq \deltaMin$ if $\deltaMin < \delta^*$.
\end{itemize}
In particular, if $\delta^* = 1$, then it implies that $\exists$ a node $w\in N$ in $\mathcal{G}^*$, where all simple paths starting at $v$ and terminating at $m'$ must traverse $w$, i.e., $\nexists$ simple monitor-to-monitor paths traversing $v$ ($P_v=\emptyset$); therefore $\OmegaCSP(v)=0$ (even single-node failures in $S$ cannot always be localized if $v\in S$).

The only cases when  $\pi_v \leq \sigma-2$ is violated are: (i) $\deltaMin = \delta^* = \sigma$, or (ii) $\deltaMin = \sigma-1$ and $\delta^* = \sigma$. In case (i), non-monitor $v$ still has a monitor as a neighbor after removing $m$; by {\clr Proposition}~\ref{prop:sigma-identifiability, CSP}, this implies that $\OmegaCSP(v) = \sigma$. In case (ii), Theorem~\ref{thm:k-identifiability, CSP}~(a) can still be applied to show that $\OmegaCSP(S) \geq \sigma-2$, and one can verify that the condition in {\clr Proposition}~\ref{prop:sigma-identifiability, CSP} is violated, which implies that $\OmegaCSP(v) \leq \sigma-1$. In fact, we can leverage {\clr Proposition}~\ref{prop:(sigma-1)-identifiability, CSP} to uniquely determine $\OmegaCSP(S)$ in this case. If the conditions in {\clr Proposition}~\ref{prop:(sigma-1)-identifiability, CSP} are satisfied, then $\OmegaCSP(v) = \sigma-1$; otherwise, $\OmegaCSP(v) = \sigma-2$.


\textbf{Evaluation algorithm:}
Evaluating $\OmegaCSP(S)$ by Proposition~\ref{prop:Construct_Omega_S} involves computing $\Omega(v)$ for all $v\in S$, each requiring the computation of the vertex-cuts of the auxiliary graphs $\mathcal{G}^*$ and $\mathcal{G}_m$ ($\forall m\in M$) as that in Section~\ref{subsec:Conditions for k-identifiability, CAP}, which altogether takes $O(\mu \theta \xi |S|)$ time.

\subsection{Maximum Identifiability Index under UP}

As in the case of CSP, we can leverage the sufficient and the necessary conditions in Theorem~\ref{thm:k-identifiability, UP} to bound the maximum identifiability index under UP from both sides. The conditions in Theorem~\ref{thm:k-identifiability, UP} imply the following bounds on the maximum identifiability index under UP.

\begin{theorem}[Maximum Per-node Identifiability under UP]\label{thm:Omega_UP}
The maximum identifiability index of a non-monitor $v$ under UP with measurement paths $P$ is bounded by $\MSC(v) - 1 \leq \OmegaUP(v) \leq \MSC(v)$.
\end{theorem}

\begin{proof}
The proof can be found in \cite{MaTONTR_Mar2015}.
\end{proof}


\textbf{Evaluation algorithm:}
The original bounds in Theorem~\ref{thm:Omega_UP} are hard to evaluate due to the NP-hardness of computing $\MSC(\cdot)$. As in Section~\ref{subsec:Conditions for k-identifiability, UP}, we resort to the greedy algorithm, which implies the following relaxed bounds:
\begin{align}\label{eq:relaxed bounds on Omega_UP}
\Big\lceil {\frac{\GSC(v)}{\log(|P_v|)+1}} \Big\rceil - 1 \leq \OmegaUP(v) \leq \GSC(v).
\end{align}
Evaluating these bounds for $\OmegaUP(S)$ involves invoking the greedy algorithm for each node in $S$, with an overall complexity of $O(|S||P|^2\sigma)$ (or $O(\mu^4 \sigma|S|)$ if all monitors can probe each other).

\section{Characterization of the Maximum Identifiable Set}
\label{sec:CharacterizationMaximumIdentifiableSet}

By Proposition~\ref{prop:construct_maximum_identifiable_set}, the maximum $k$-identifiable set $S^*(k)$ is related to the per-node maximum identifiability index $\Omega(v)$ by $S^*(k) = \{v\in N: \Omega(v)\geq k\}$. Therefore, $S^*(k)$ can be easily computed based on values of $\Omega(v)$ ($v\in N$) for any value of $k$. Moreover, given upper/lower bounds on $\Omega(v)$, i.e., $\Omega_l(v)\leq \Omega(v)\leq \Omega_u(v)$, $S^*(k)$ can be bounded by $\Sinner(k) \subseteq S^*(k) \subseteq \Souter(k)$ for $\Sinner(k) := \{v\in N: \Omega_l(v)\geq k\}$ and $\Souter(k) := \{v\in N: \Omega_u(v) \geq k\}$. Based on this observation, we now characterize $S^*(k)$ for each of the three probing mechanisms.

\subsection{Maximum $k$-identifiable Set under CAP}
\label{subsec:MaximumIdentifiableSetUnderCAP}

The expression of the maximum per-node identifiability under CAP in Theorem~\ref{thm:Omega_CAP} leads to the following characterization of the maximum $k$-identifiable set.

\begin{corollary}
\label{coro:Maximum_k_identifiable_set-CAP}
The maximum $k$-identifiable set under CAP, denoted by $\SCAP(k)$, is $\SCAP(k)=\{v\in N: \Gamma_{\mathcal{G}^*}(v,m')\geq k\}$.
\end{corollary}

Specifically, when $k=\sigma$, $\SCAP(\sigma)$ contains all the non-monitors directly adjacent to monitors.

\textbf{Evaluation algorithm:}
As shown in Section~\ref{subsec:Conditions for k-identifiability, CAP}, $\Gamma_{\mathcal{G}^*}(v,m')$ can be computed in $O(\theta \xi)$ time. Thus, the total time complexity for constructing $\SCAP(k)$ is $O(\theta \xi \sigma)$.

\subsection{Maximum $k$-identifiable Set under CSP}
\label{subsec:MaximumIdentifiableSetUnderCSP}

Leveraging Theorem~\ref{thm:Omega_CSP}, we can establish outer and inner bounds (i.e., superset and subset) for the maximum $k$-identifiable set under CSP.

\begin{corollary}
\label{coro:Maximum_k_identifiable_set-CSP}
Let $\SCSPo(k):=\{v \in N:\pi_v\geq k\}$, and $\SCSPi(k):=\{v \in N:\pi_v\geq k+1\}$. The maximum $k$-identifiable set under CSP ($k\leq \sigma-1$), denoted by $\SCSP(k)$, is bounded by $\SCSPi(k)\subseteq\SCSP(k)\subseteq\SCSPo(k)$.
\end{corollary}
\begin{proof}
The proof can be found in \cite{MaTONTR_Mar2015}.
\end{proof}


One case not covered by Corollary~\ref{coro:Maximum_k_identifiable_set-CSP} is $k=\sigma$. In this case, $\SCSP(\sigma)$ contains all non-monitors that have at least two monitors as neighbors according to Proposition~\ref{prop:sigma-identifiability, CSP}. Another non-covered case is $k=\sigma-1$, for which we have the following result.\looseness=-1

\begin{corollary}
\label{coro:Maximum_sigma-1_identifiable_set-CSP}
When $k=\sigma-1$, $\SCSP(k)=\{v\in N:v$ has at least two monitor neighbors$\}\cup \widetilde{S}$. Set $\widetilde{S}$ contains one and only one non-monitor $w$ if all nodes in $N$ but $w$ have at least two monitor neighbors and $w$ has one monitor and all nodes in $N\setminus \{w\}$ as neighbors; otherwise, $\widetilde{S}=\emptyset$.
\end{corollary}
\begin{proof}
The proof can be found in \cite{MaTONTR_Mar2015}.
\end{proof}

Corollary~\ref{coro:Maximum_sigma-1_identifiable_set-CSP} implies that when $\widetilde{S}$ is not empty (i.e., $|\widetilde{S}|=1$), then $\SCSP(\sigma-1)=N$ and $\SCSP(\sigma)=N\setminus \widetilde{S}$ (i.e., $|\SCSP(\sigma-1)|=\sigma$ and $|\SCSP(\sigma)|=\sigma-1$).

\textbf{Evaluation algorithm:}
Corollary~\ref{coro:Maximum_sigma-1_identifiable_set-CSP} is computable in linear time. Similar to Section~\ref{subsec:Conditions for k-identifiability, CSP}, $\pi_v$ in Corollary~\ref{coro:Maximum_k_identifiable_set-CSP} is in $O(\mu \theta \xi)$ complexity. Therefore, the overall complexity is $O(\mu \theta \xi \sigma)$.

\subsection{Maximum $k$-identifiable Set under UP}
\label{subsec:MaximumIdentifiableSetUnderUP}

Analogous to the case of CSP, we leverage Theorem~\ref{thm:Omega_UP} to develop the following outer and inner bounds for the maximum $k$-identifiable set under UP.

\begin{corollary}
\label{coro:Maximum_k_identifiable_set-UP}
Let $\SUPo(k):=\{v \in N:\MSC(v)\geq k\}$ and $\SUPi(k):=\{v \in N:\MSC(v)\geq k+1\}$ with measurement paths $P$. The maximum $k$-identifiable set under UP ($k\leq \sigma-1$), denoted by $\SUP(k)$, is bounded by $\SUPi(k)\subseteq\SUP(k)\subseteq\SUPo(k)$.
\end{corollary}
\begin{proof}
The proof can be found in \cite{MaTONTR_Mar2015}.
\end{proof}

A special case left out by Corollary~\ref{coro:Maximum_k_identifiable_set-UP} is $k=\sigma$. In this case, we use Proposition~\ref{prop:sigma-identifiability_UP} to determine $\SUP(\sigma)$, i.e., $\SUP(\sigma)=\{w\in N: w$ is on a $2$-hop path$\}$.

\textbf{Evaluation algorithm:}
Due to the NP-hardness of computing $\MSC(\cdot)$, we again resort to the greedy algorithm, whereby the outer and inner bounds of $\SUP(k)$ can be relaxed by computing $\GSC(\cdot)$. Let $\widehat{\SUPo}(k):=\{v \in N:\GSC(v)\geq k\}$ and $\widehat{\SUPi}(k):=\{v \in N:\GSC(v)/\big(\log(|P_v|)+1\big)\geq k+1\}$. We have $\SUPo(k)\subseteq\widehat{\SUPo}(k)$ and $\widehat{\SUPi}(k)\subseteq\SUPi(k)$ according to Proposition~\ref{prop:construct_maximum_identifiable_set}. The computation of these relaxed bounds involves $O(\sigma|P|^2)$ time complexity w.r.t. each node in $N$. Thus, the overall complexity is $O(\sigma^2|P|^2)$.

\section{Evaluation of Failure Localization Capability}\label{sec:Impact of Probing Mechanisms}

We demonstrate how the proposed measures of maximum identifiability index and maximum identifiable set can be used to evaluate the impact of various parameters, including topology, number of monitors, and probing mechanisms (CAP, CSP, UP), on the capability of failure localization.
In this study, we assume (hop count-based) shortest path routing as the default routing protocol under UP, i.e., measurement paths under UP are the shortest paths between monitors, with ties broken arbitrarily.
\begin{figure}[tb]
\vspace{-.7em}
\begin{minipage}{.5\linewidth}
  \centerline{\includegraphics[width=1.05\columnwidth]{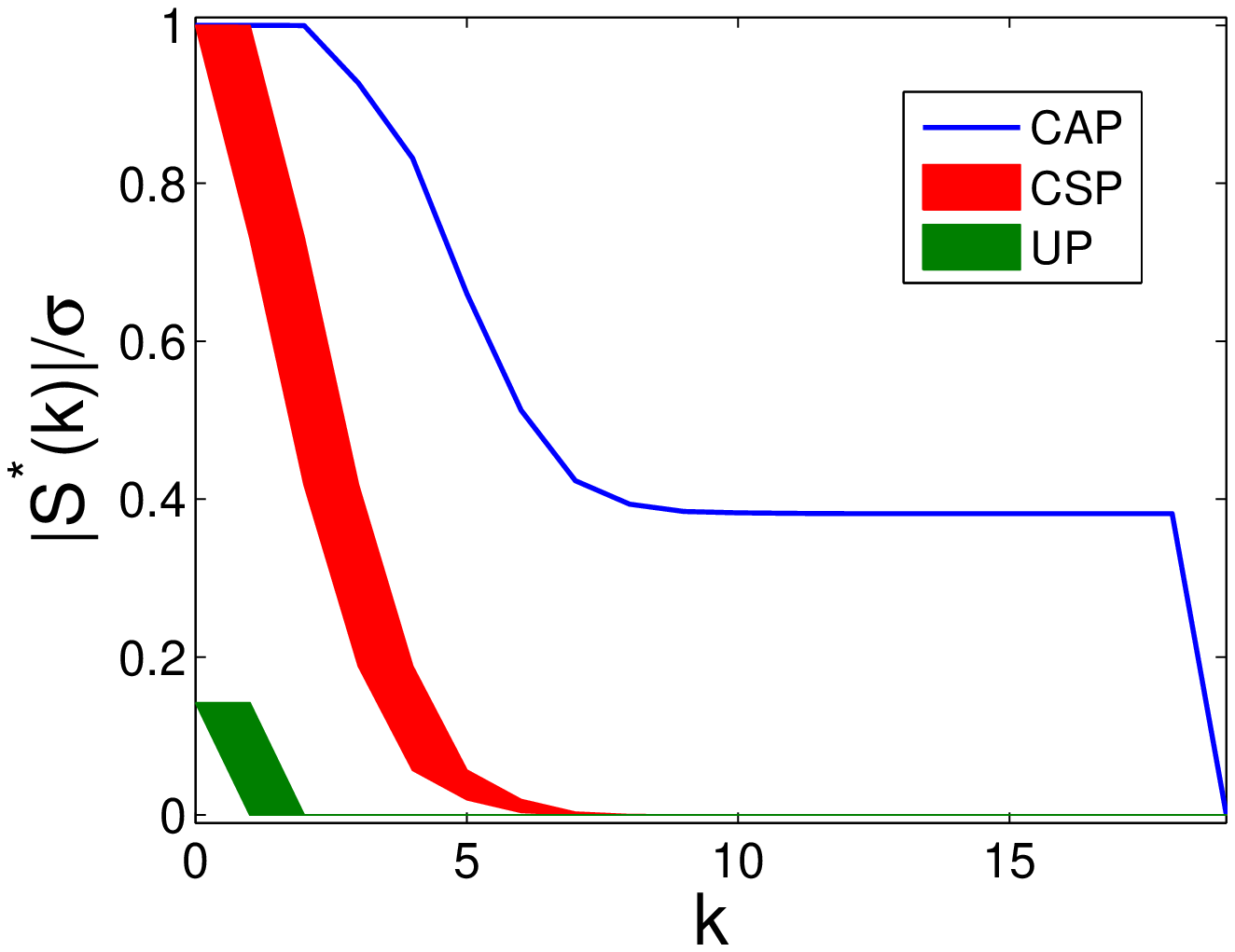}}
  \vspace{-.5em}\centerline{\scriptsize (a) $\mu=2$}
\end{minipage}\hfill
\begin{minipage}{.5\linewidth}
  \centerline{\includegraphics[width=1.05\columnwidth]{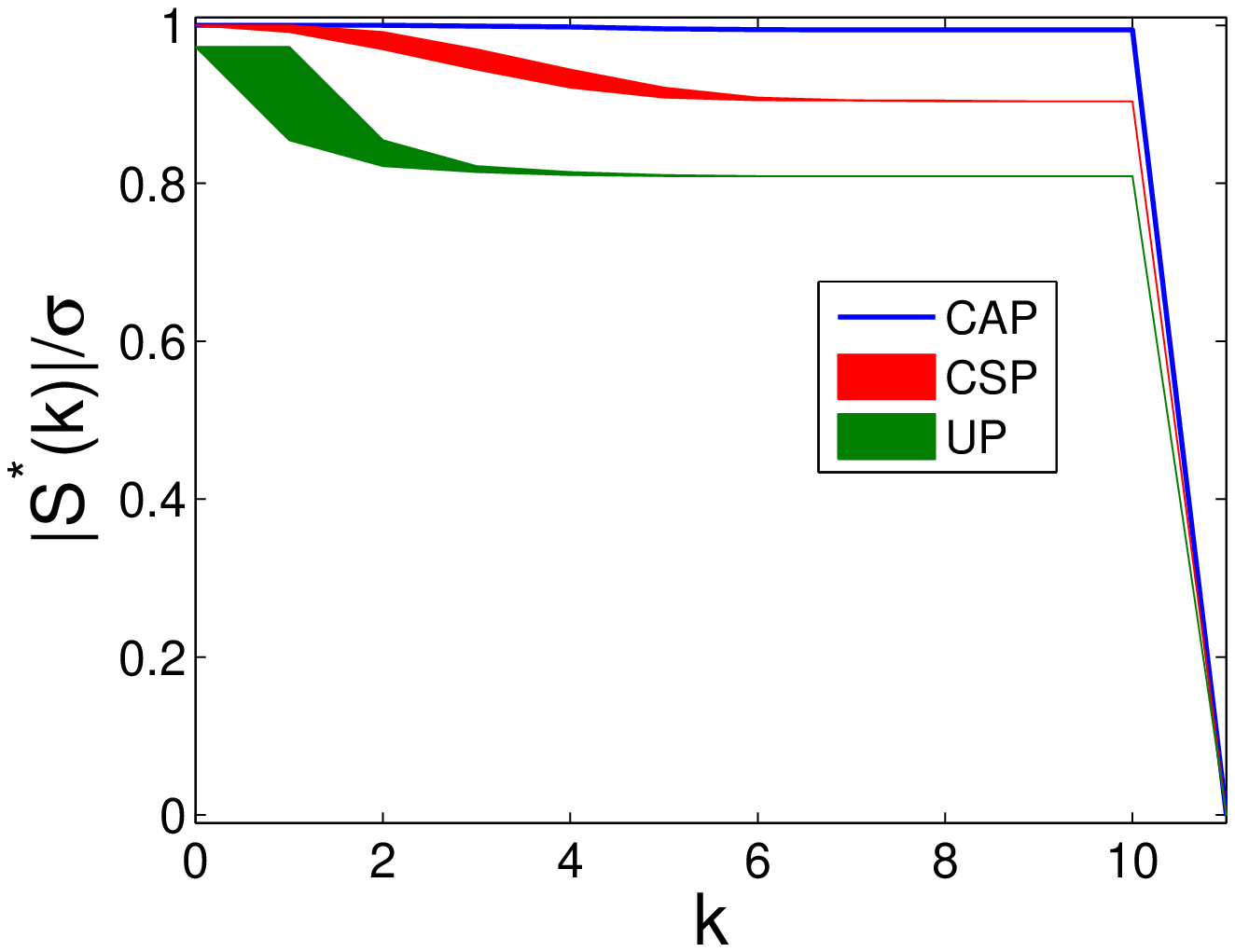}}
  \vspace{-.5em}\centerline{\scriptsize (b) $\mu=10$}
\end{minipage}
\caption{Maximum $k$-identifiable set $S^*(k)$ under CAP, CSP, and UP for ER graphs ($|V|=20$, $\mu=\{2,10\}$, $\mathbb{E}[|L|]=51$, $200$ graph instances, $\sigma$: total number of non-monitors). } \label{fig:S_bounds_ER}
\vspace{-0.5em}
\end{figure}

\begin{figure}[tb]
\vspace{-.0em}
\begin{minipage}{.5\linewidth}
  \centerline{\includegraphics[width=1.05\columnwidth]{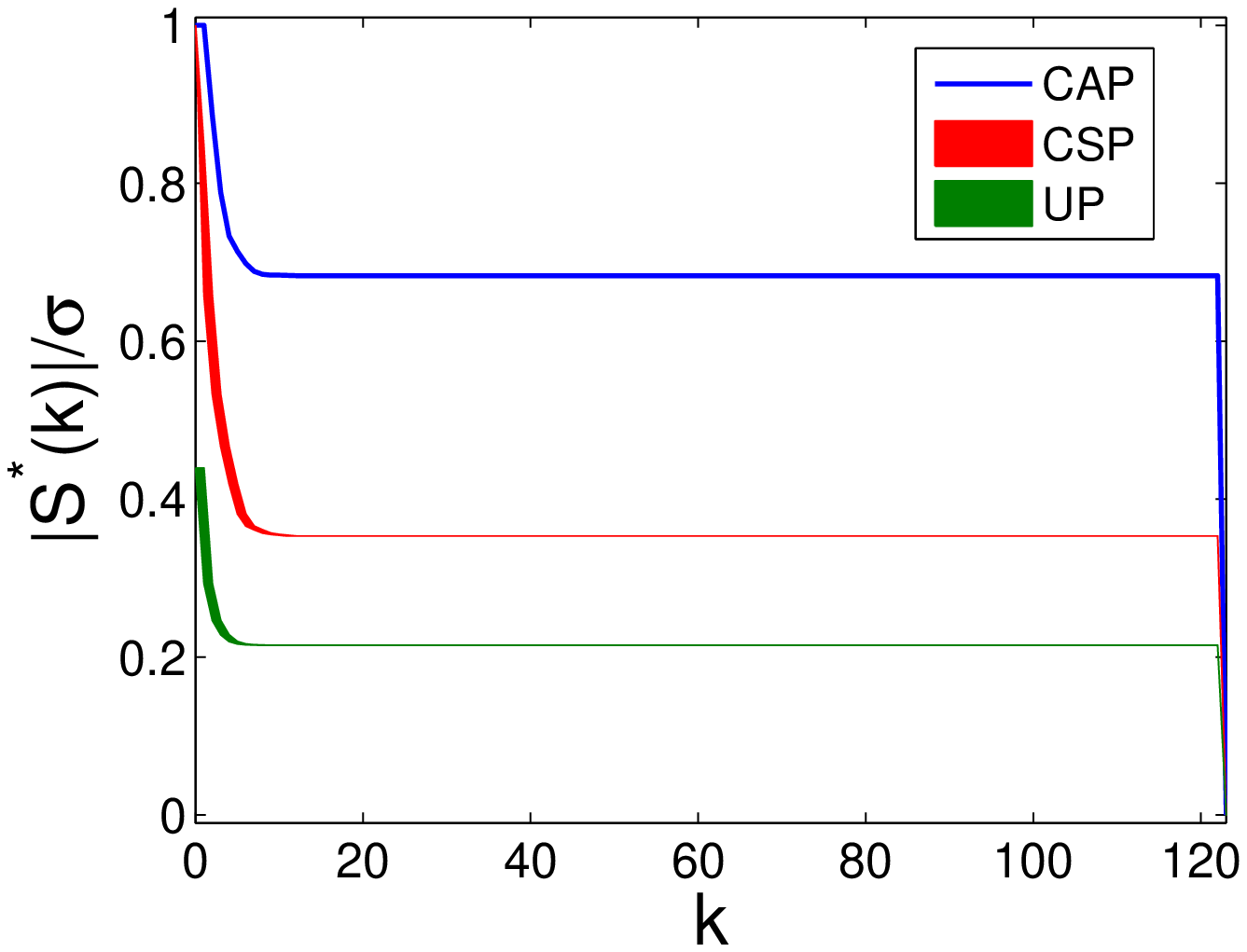}}
  \vspace{-.5em}\centerline{\scriptsize (a) $\mu=50$}
\end{minipage}\hfill
\begin{minipage}{.5\linewidth}
  \centerline{\includegraphics[width=1.05\columnwidth]{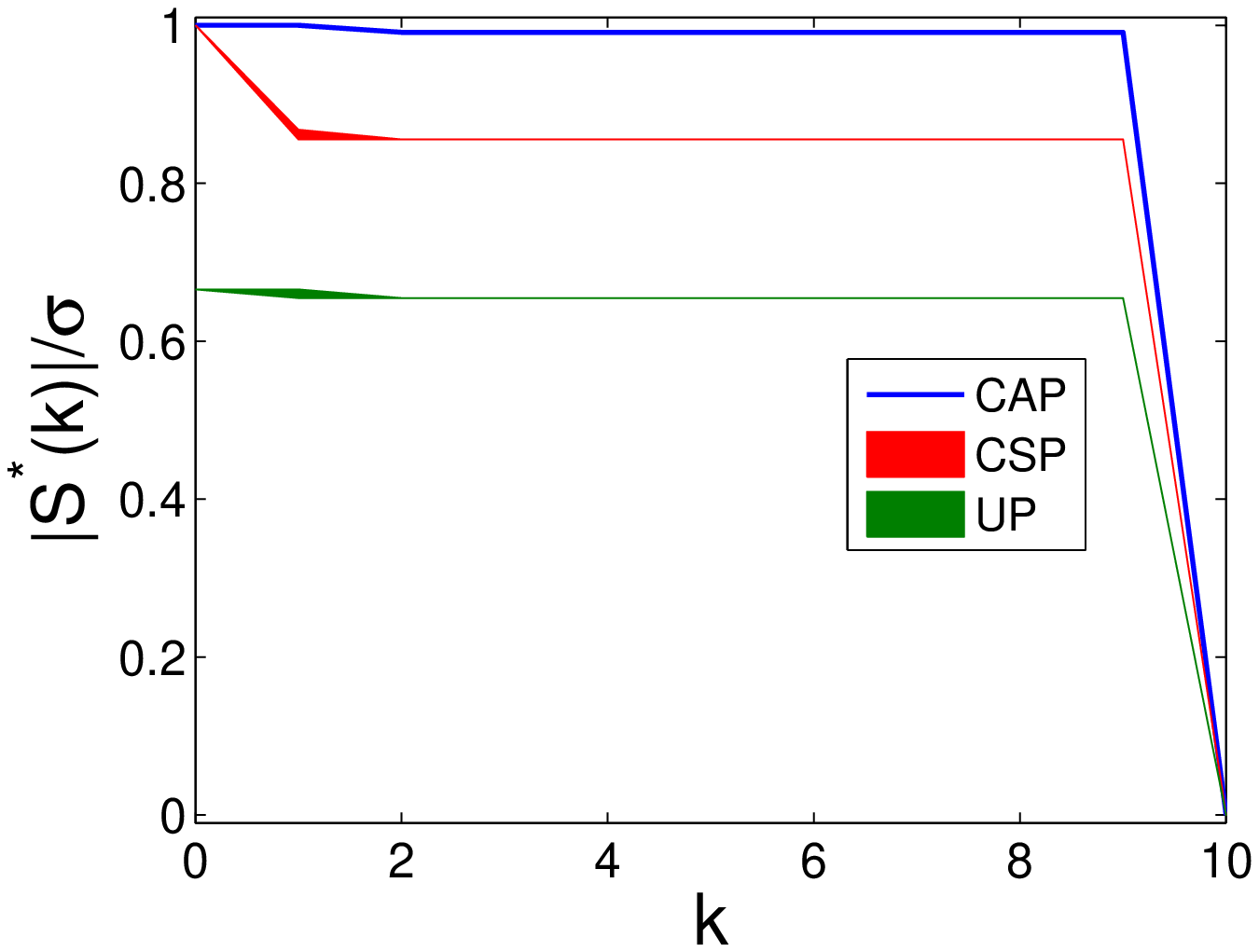}}
  \vspace{-.5em}\centerline{\scriptsize (b) $\mu=163$}
\end{minipage}
\vspace{-.0em}
\caption{Maximum $k$-identifiable set $S^*(k)$ under CAP, CSP, and UP for Rocketfuel AS1755 ($|V|=172$, $|L|=381$, $\mu=\{50,163\}$, $100$ Monte Carlo runs, $\sigma$: total number of non-monitors). } \label{fig:S_bounds_Rocketfuel}
\end{figure}

\subsection{Topologies for Evaluation}
\label{sect:randomTopo}

We first employ random graph models to generate a comprehensive set of topologies without artifacts of specific network deployments. We consider random Erd\"{o}s-R\'{e}nyi (ER) graphs \cite{ErdosRenyi60}, generated by independently connecting each pair of nodes by a link with a fixed probability $p$. The result is a purely random topology where all graphs with an equal number of links are equally likely to be selected (note that the number of nodes is an input parameter). In addition to ER graphs, other random graph models are also considered; the corresponding results are presented in \cite{MaTONTR_Mar2015} due to space limitation.

We then evaluate real \emph{Autonomous System} (AS) topologies collected by the Rocketfuel \cite{RocketFuel} and the CAIDA \cite{CAIDA} projects, which represents IP-level connections between backbone/gateway routers of several ASes from major \emph{Internet Service Providers (ISPs)} around the globe.

\subsection{Evaluation Results}

We focus on evaluating per-node maximum identifiability index $\Omega(v)$ since it determines both the per-set maximum identifiability index $\Omega(S)$ and the maximum identifiable set $S^*(k)$. In particular, the \emph{complementary cumulative distribution function (CCDF)} of $\Omega(v)$ over all $v\in N$ (refer to Table~\ref{t notion} for notations) coincides with the normalized cardinality of the maximum identifiable set $|S^*(k)|/\sigma$, and thus we characterize the distribution of $\Omega(v)$ by evaluating $|S^*(k)|/\sigma$ wrt $k$. Moreover, we examine the specific value of $\Omega(v)$ and compare it with the degree (i.e., number of neighbors) of $v$ among monitor/non-monitor nodes to evaluate the correlation between maximum identifiability index and the graph property (i.e., degree) of a node.
Under UP, our extensive simulations under multiple graph models have shown that $\MSC(v)$ can be closely approximated by $\GSC(v)$; hence, we use $\GSC(v)$ in place of $\MSC(v)$ for computing $\OmegaUP$ and $\SUP$; see \cite{MaTONTR_Mar2015} for details.\looseness=-1

\subsubsection{Distribution of $\Omega(v)$}

\begin{figure}[tb]
\vspace{-.7em}
\begin{minipage}{.5\linewidth}
  \centerline{\includegraphics[width=1.05\columnwidth]{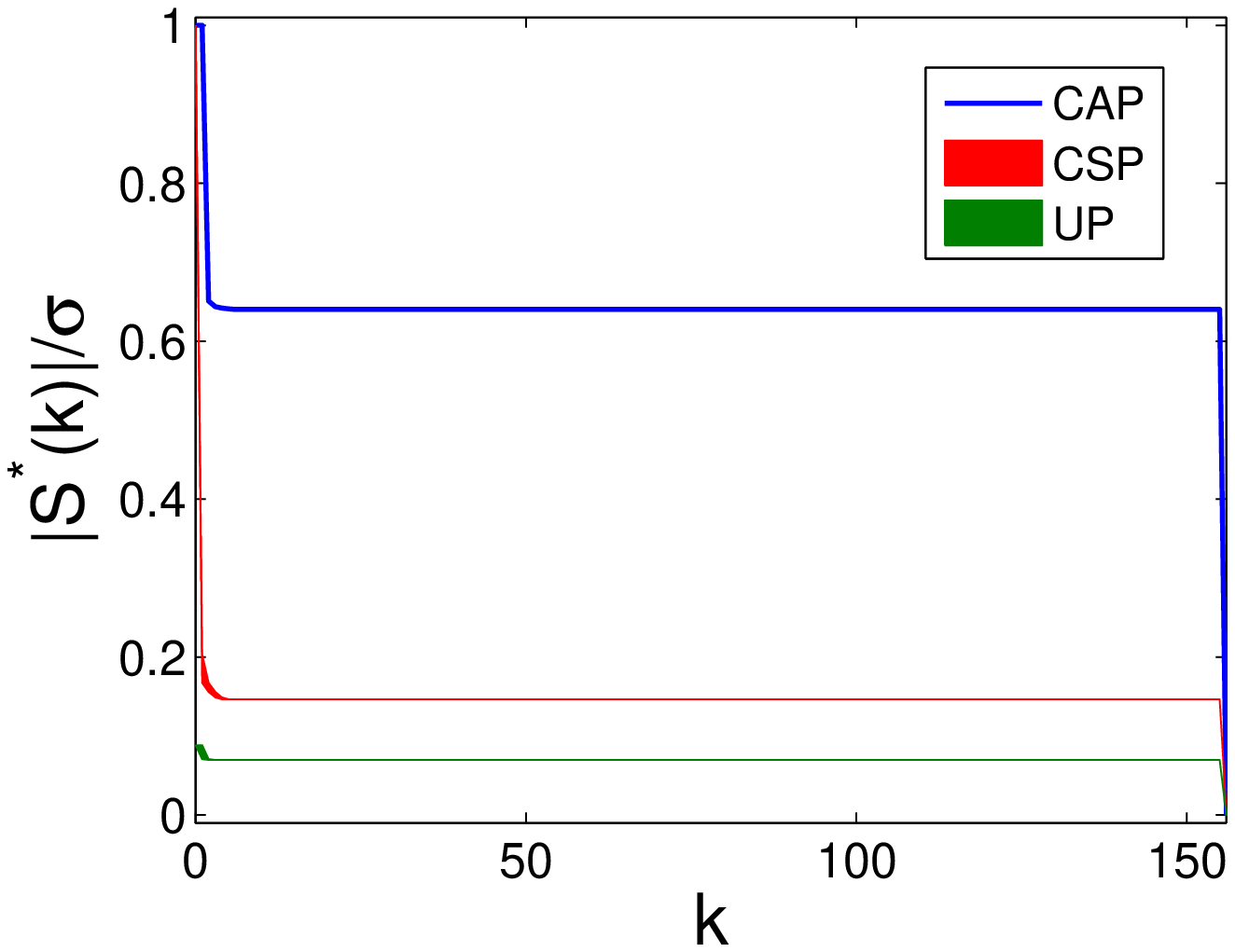}}
  \vspace{-.5em}\centerline{\scriptsize (a) $\mu=200$}
\end{minipage}\hfill
\begin{minipage}{.5\linewidth}
  \centerline{\includegraphics[width=1.05\columnwidth]{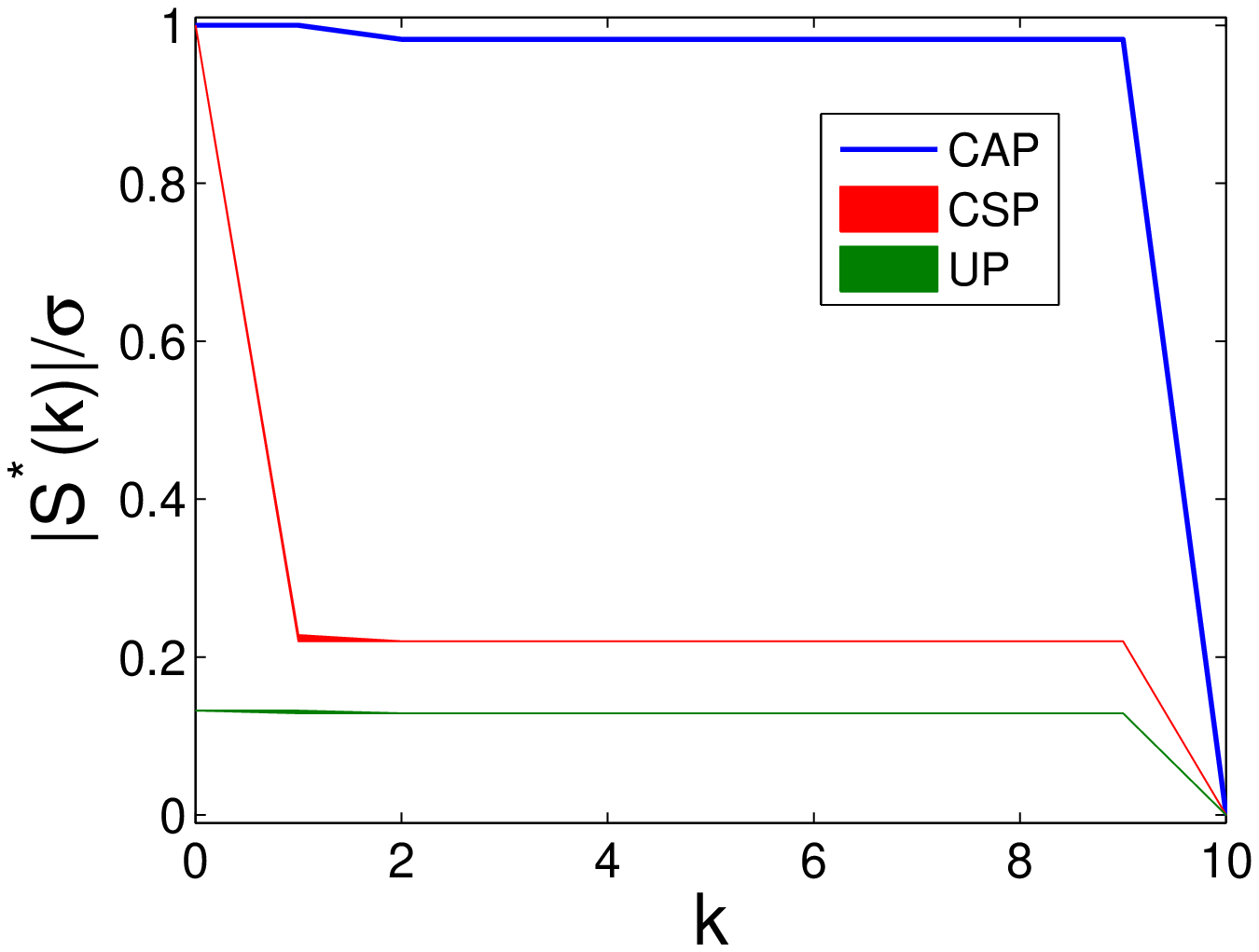}}
  \vspace{-.5em}\centerline{\scriptsize (b) $\mu=346$}
\end{minipage}
\caption{Maximum $k$-identifiable set $S^*(k)$ under CAP, CSP, and UP for CAIDA AS26788 ($|V|=355$, $|L|=483$, $\mu=\{200,346\}$, $100$ Monte Carlo runs, $\sigma$: total number of non-monitors). } \label{fig:S_bounds_CAIDA}
\vspace{-1.5em}
\end{figure}

To characterize the overall distribution of $\Omega(v)$, we compute (bounds on)\footnote{Propositions~\ref{prop:sigma-identifiability, CSP}, Corollary~\ref{coro:Maximum_sigma-1_identifiable_set-CSP}, and Proposition~\ref{prop:sigma-identifiability_UP} are used to determine the exact elements in $\SCSP(\sigma)$, $\SCSP(\sigma-1)$, and $\SUP(\sigma)$.} $\SCAP(k),\: \SCSP(k)$, and $\SUP(k)$ to evaluate $|S^*(k)|/\sigma$ for different values of $k$ ($\sigma$: total number of non-monitors). Fig.~\ref{fig:S_bounds_ER} reports averages of $|S^*(k)|/\sigma$ computed on ER graphs over randomly selected multiple instances of topology and monitor locations, where $|S^*(k)|/\sigma$ under CSP and UP is represented by a band with its width determined by $(|S^{\mbox{\tiny outer}}(k)|-|S^{\mbox{\tiny inner}}(k)|)/\sigma$.
The results show large differences in the failure localization capabilities of different probing mechanisms: When the number of monitors is small ($\mu = 2$) and $k=2$, $\SUP(k)$ is almost empty, i.e., no (non-monitor) node state can be uniquely determined by UP when there are multiple failures; in contrast, $|\SCSP(k)|/\sigma\approx 0.5$ and $|\SCAP(k)|/\sigma\approx 1$, i.e., CSP can uniquely determine the states of half of the nodes and CAP can determine the states of all the nodes when $\mu = 2$ and $k=2$.  When the number of monitors increases ($\mu = 10$), there exist more measurement paths between monitors, and thus the fraction of identifiable nodes increases for all three probing mechanisms.
In addition, we observe a stable phase in Fig.~\ref{fig:S_bounds_ER} where the value of $|S^*(k)|/\sigma$ remains the same as we increase $k$; this is because some non-monitors have monitors as neighbors, thus directly measurable by these neighboring monitors without traversing other non-monitors. Specifically, if there are non-monitors that neighbor at least one monitor under CAP, neighbor at least two monitors under CSP, or lie on 2-hop paths between monitors under UP, then the failure of these non-monitors can always be identified regardless of the total number of failures in the network, i.e., the maximum identifiability index of these non-monitors is the total number of non-monitors.
Note that in Fig.~\ref{fig:S_bounds_ER}, the number of such directly measurable non-monitors is smaller under UP than under CSP. This is because for non-monitors that neighbor the same pair of monitors (e.g., $m_1$ and $m_2$), all these non-monitors are directly measurable on 2-hop $m_1$-to-$m_2$ paths under CSP; however, only one of these non-monitors is on a 2-hop $m_1$-to-$m_2$ path under UP as UP probes only one routing path between each pair of monitors (assuming stable single-path routing).
Similar results have been obtained for other random graph models (see \cite{MaTONTR_Mar2015} for details).

We repeat the above evaluation on AS topologies. We select AS1755 from Rocketfuel topologies \cite{RocketFuel} and AS26788 from CAIDA topologies \cite{CAIDA}, and evaluate the bounds on $|S^*(k)|/\sigma$ under multiple instances of random monitor placements; average results are reported in Fig.~\ref{fig:S_bounds_Rocketfuel} and \ref{fig:S_bounds_CAIDA}.  Similar to the case of random topologies, there are clear differences between different probing mechanisms.
Unlike the uniformly connected random topologies in Fig.~\ref{fig:S_bounds_ER}, these AS topologies contain many sparse subgraphs where the removal of a few nodes can disconnect the network. Thus, unless a node is directly measurable by monitors, it is likely that failures of a few other nodes will disconnect it from monitors and thus make its failure undetectable.   
Comparing results from Rocketfuel and CAIDA, we observe that the CAIDA AS requires more monitors to achieve the same level of identifiability. Moreover, deploying more monitors in CAIDA AS only slightly improves the level of identifiability. This can be explained by examining the link density $|L|/|V|$ of the network: $|L|/|V|=1.36$ for the CAIDA AS, whereas $|L|/|V|=2.22$ for the Rocketfuel AS, i.e., CAIDA AS topology is nearly a tree. Therefore, it is likely for a node to not reside on any paths between monitors or become unmeasurable after the failure of one other node in the CAIDA AS, even if the paths are controllable but cycle-free (CSP).\looseness=-1

\subsubsection{Correlation of $\Omega(v)$ and Degree}

\begin{figure*}[t]
\vspace{-.7em}
\begin{minipage}{.3\linewidth}
  \centerline{\includegraphics[width=1\columnwidth]{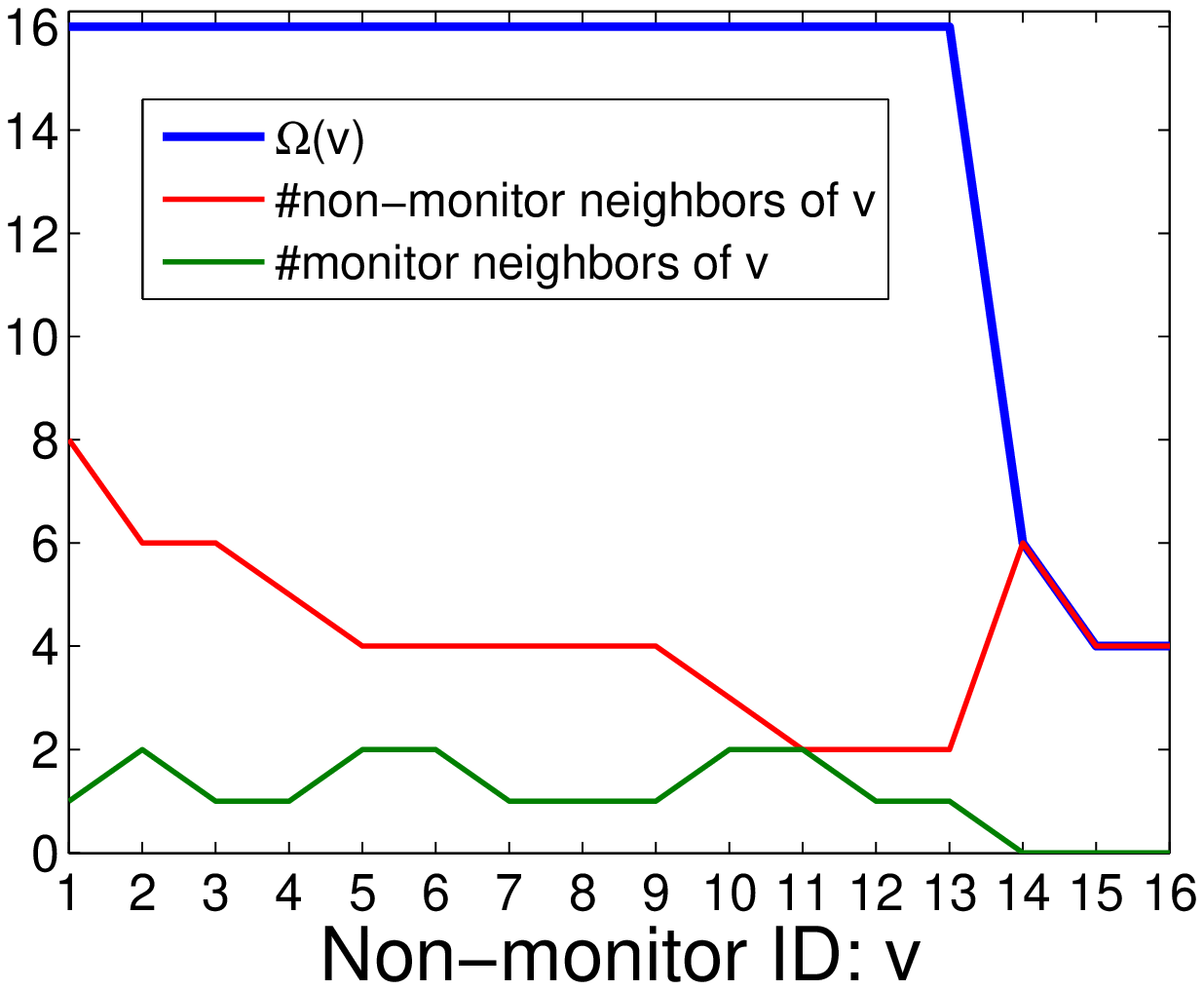}}
  \vspace{-.5em}\centerline{\small (a) Under CAP}
\end{minipage}\hfill
\begin{minipage}{.3\linewidth}
  \centerline{\includegraphics[width=1\columnwidth]{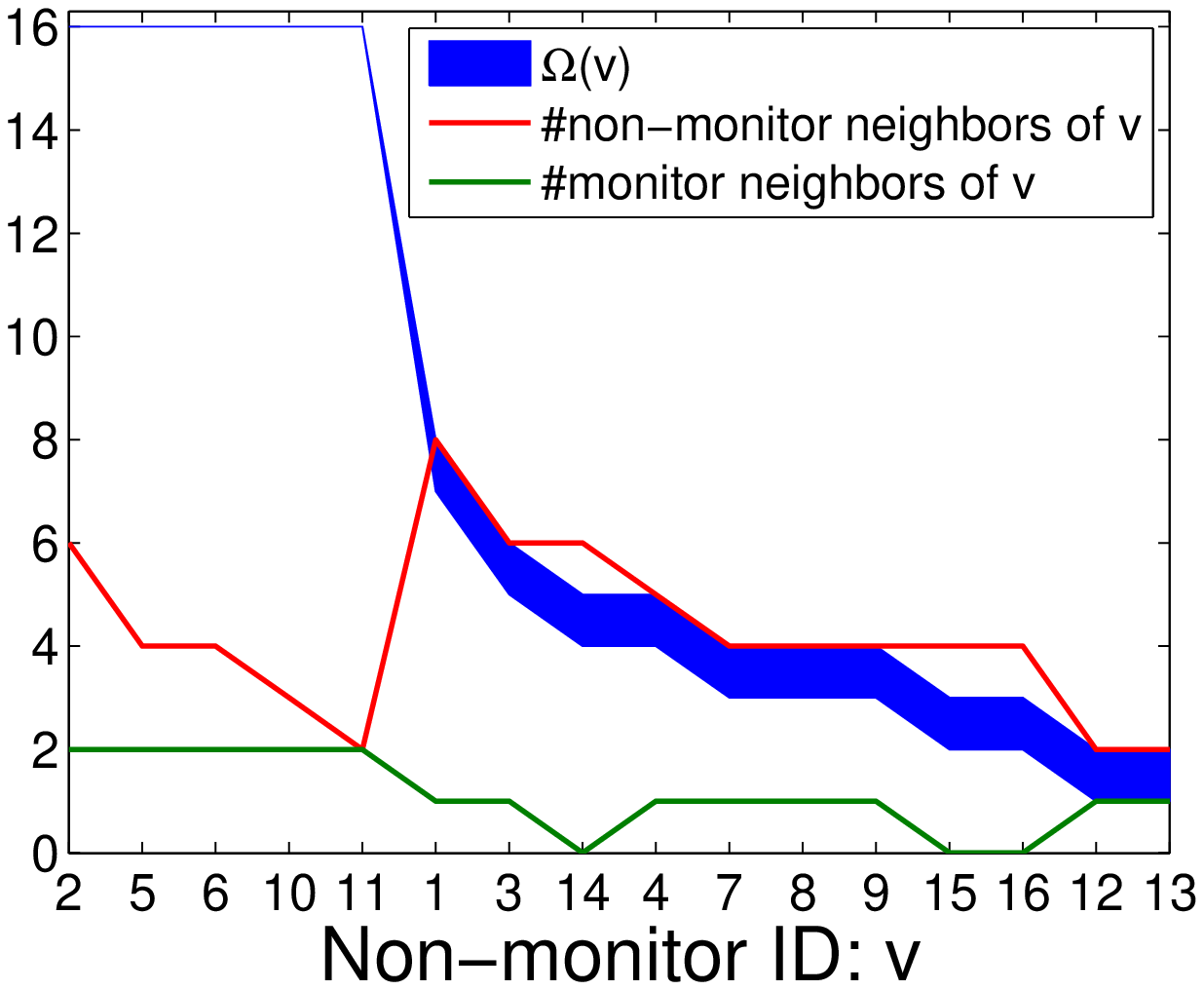}}
  \vspace{-.5em}\centerline{\small (b) Under CSP}
\end{minipage}\hfill
\begin{minipage}{.3\linewidth}
  \centerline{\includegraphics[width=1\columnwidth]{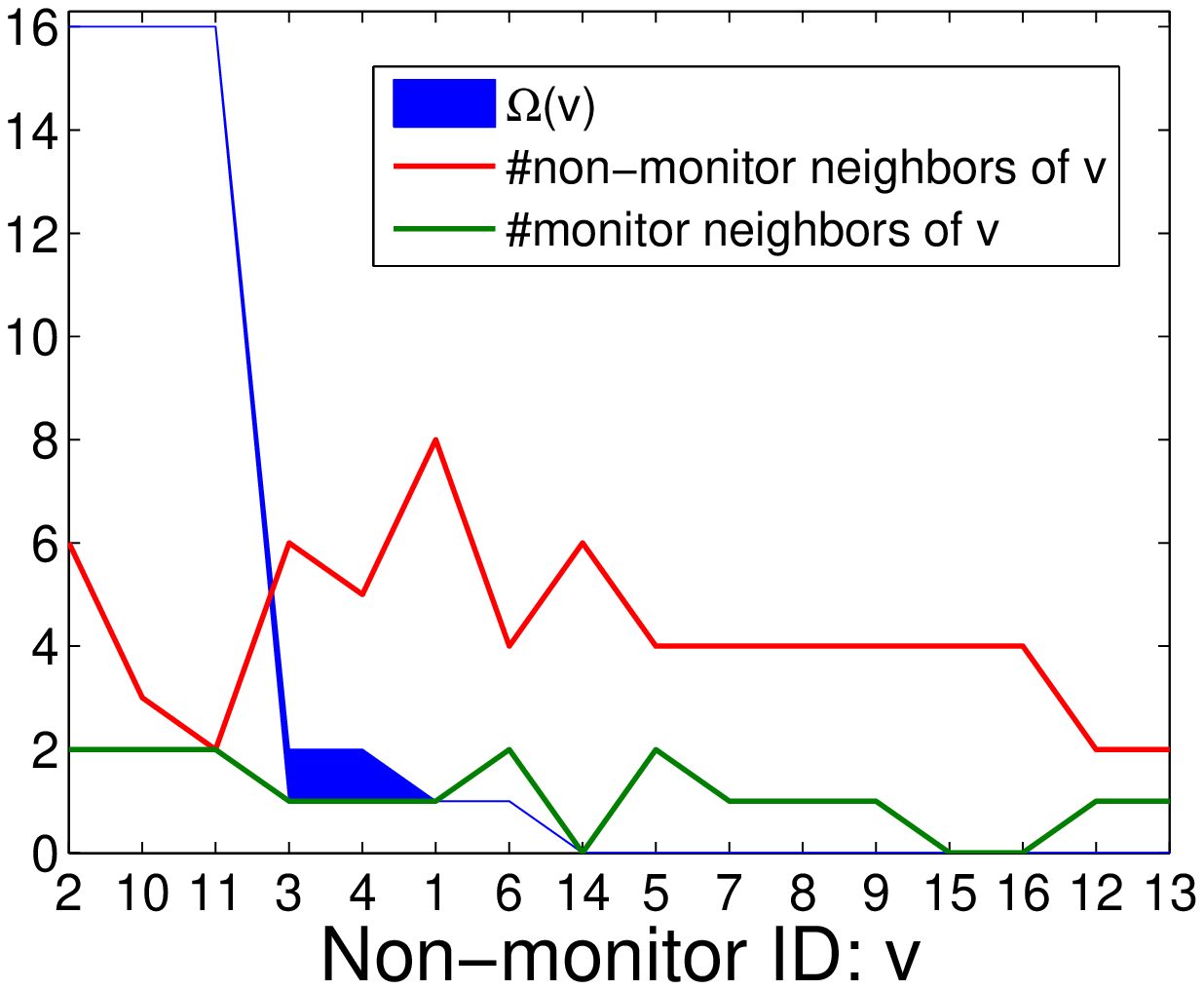}}
  \vspace{-.5em}\centerline{\small (c) Under UP}
\end{minipage}
\vspace{-.2em}
\caption{Node maximum identifiability index $\Omega(v)$ of one ER graph under different probing mechanisms ($|V|=20$, $\mu=4$, $\mathbb{E}[|L|]=51$). } \label{fig:node_maximum_identifiability}
\end{figure*}

\begin{figure*}[t]
\vspace{-.7em}
\begin{minipage}{.3\linewidth}
  \centerline{\includegraphics[width=1\columnwidth]{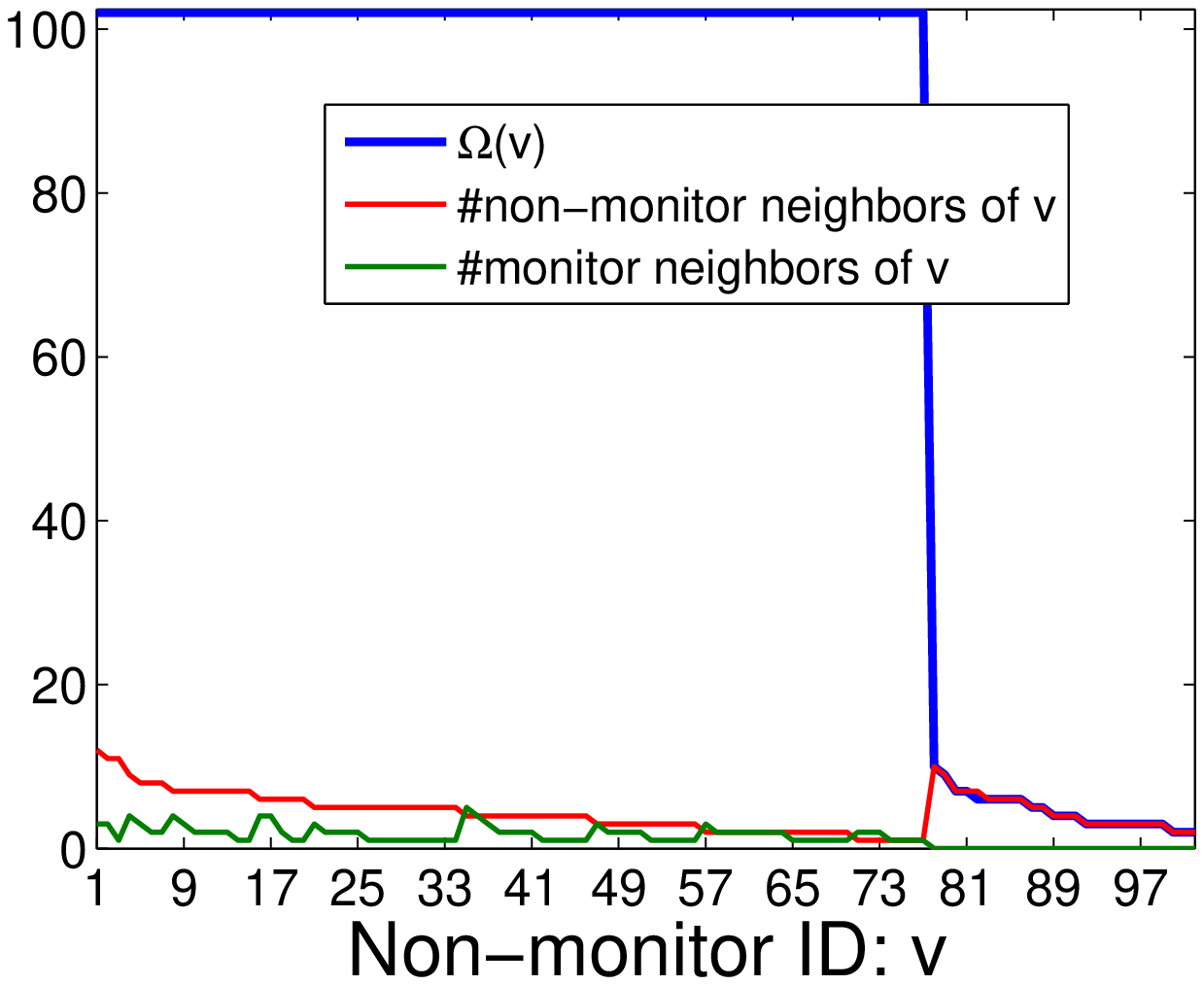}}
  \vspace{-.5em}\centerline{\small (a) Under CAP}
\end{minipage}\hfill
\begin{minipage}{.3\linewidth}
  \centerline{\includegraphics[width=1\columnwidth]{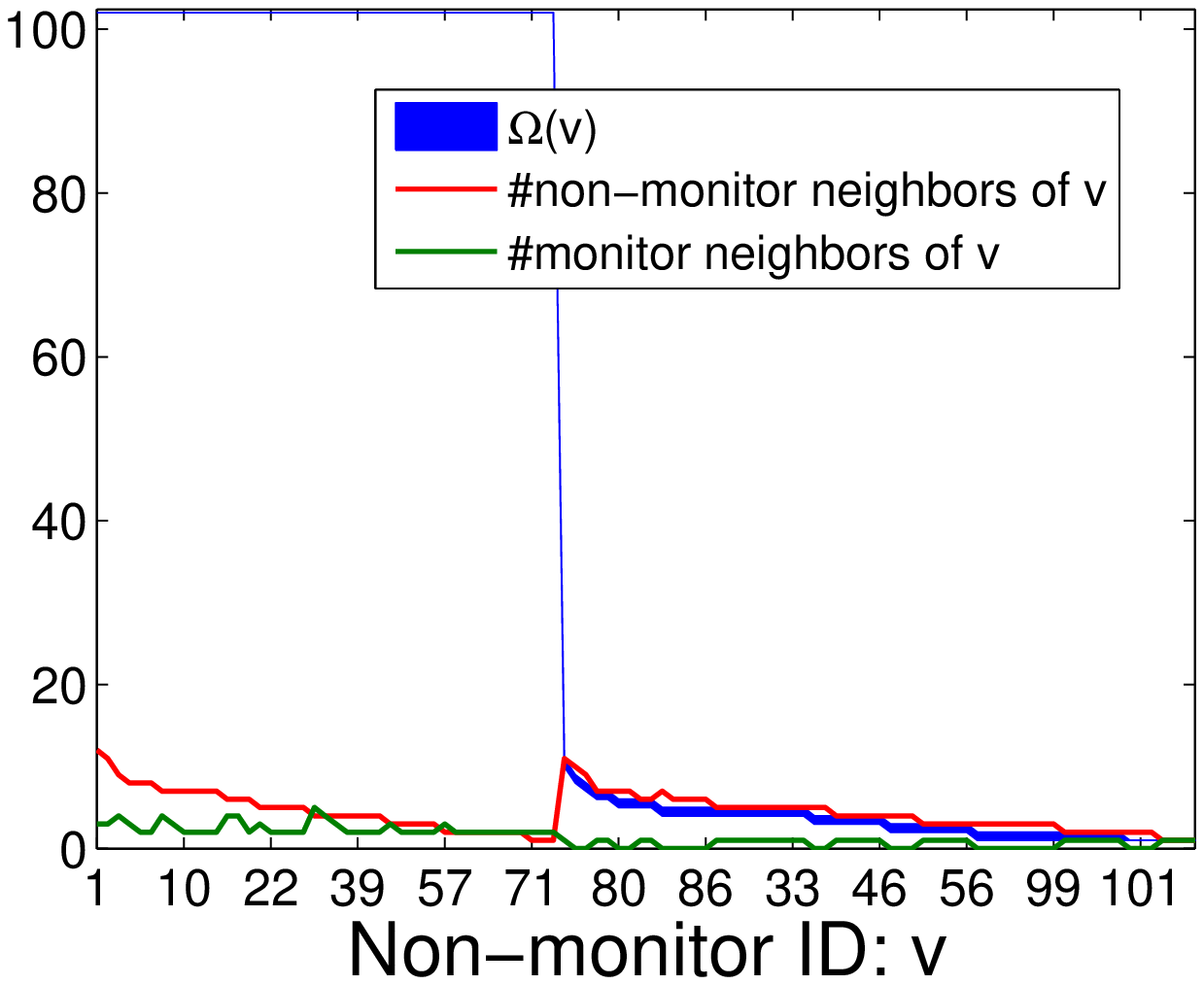}}
  \vspace{-.5em}\centerline{\small (b) Under CSP}
\end{minipage}\hfill
\begin{minipage}{.3\linewidth}
  \centerline{\includegraphics[width=1\columnwidth]{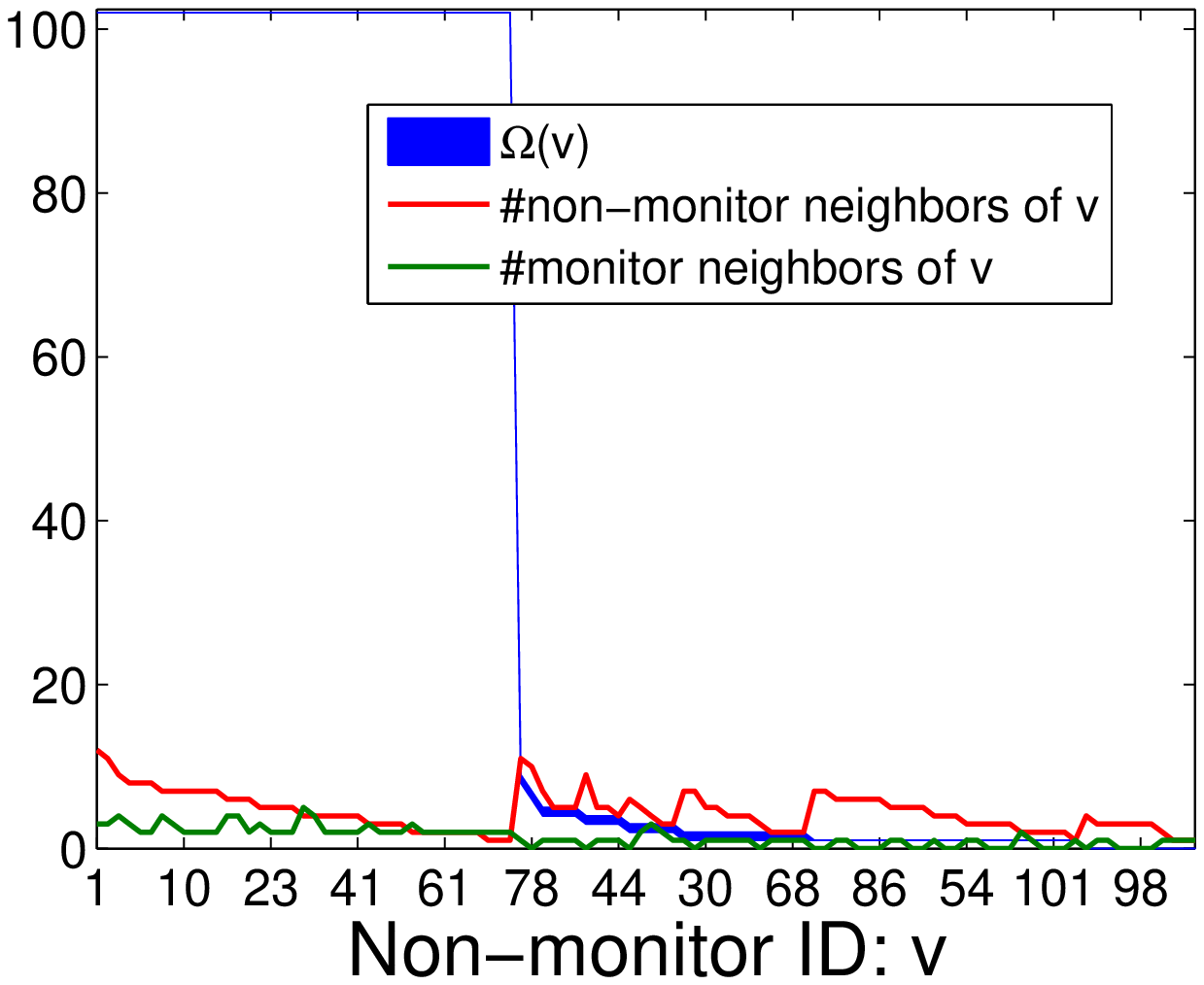}}
  \vspace{-.5em}\centerline{\small (c) Under UP}
\end{minipage}
\vspace{-.2em}
\caption{Node maximum identifiability index $\Omega(v)$ of Rocketfuel AS1755 under different probing mechanisms ($|V|=172$, $|L|=381$, $\mu=70$). } \label{fig:node_maximum_identifiability_Rocketfuel}
\end{figure*}

\begin{figure*}[t]
\vspace{-.7em}
\begin{minipage}{.3\linewidth}
  \centerline{\includegraphics[width=1\columnwidth]{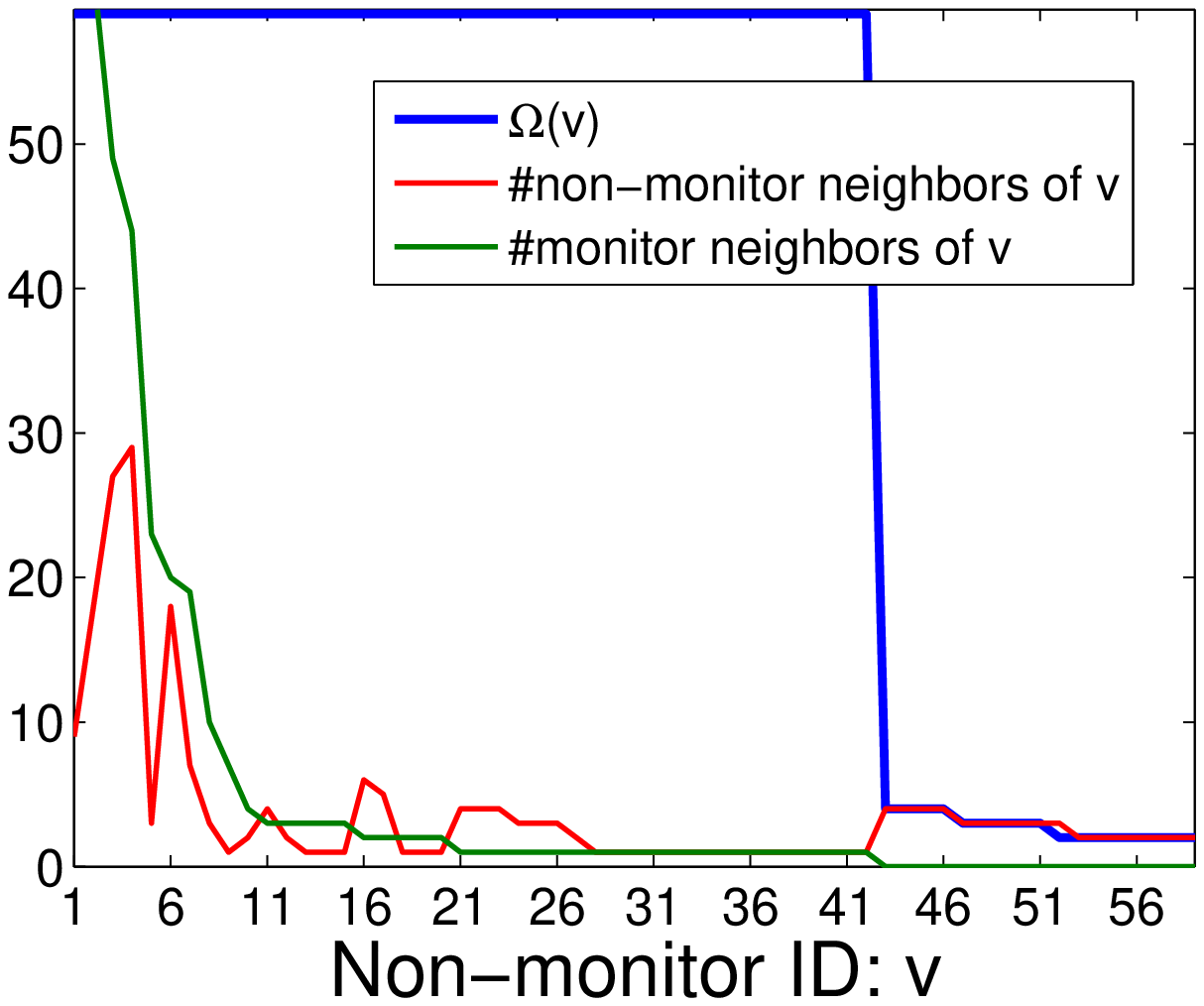}}
  \vspace{-.5em}\centerline{\small (a) Under CAP}
\end{minipage}\hfill
\begin{minipage}{.3\linewidth}
  \centerline{\includegraphics[width=1\columnwidth]{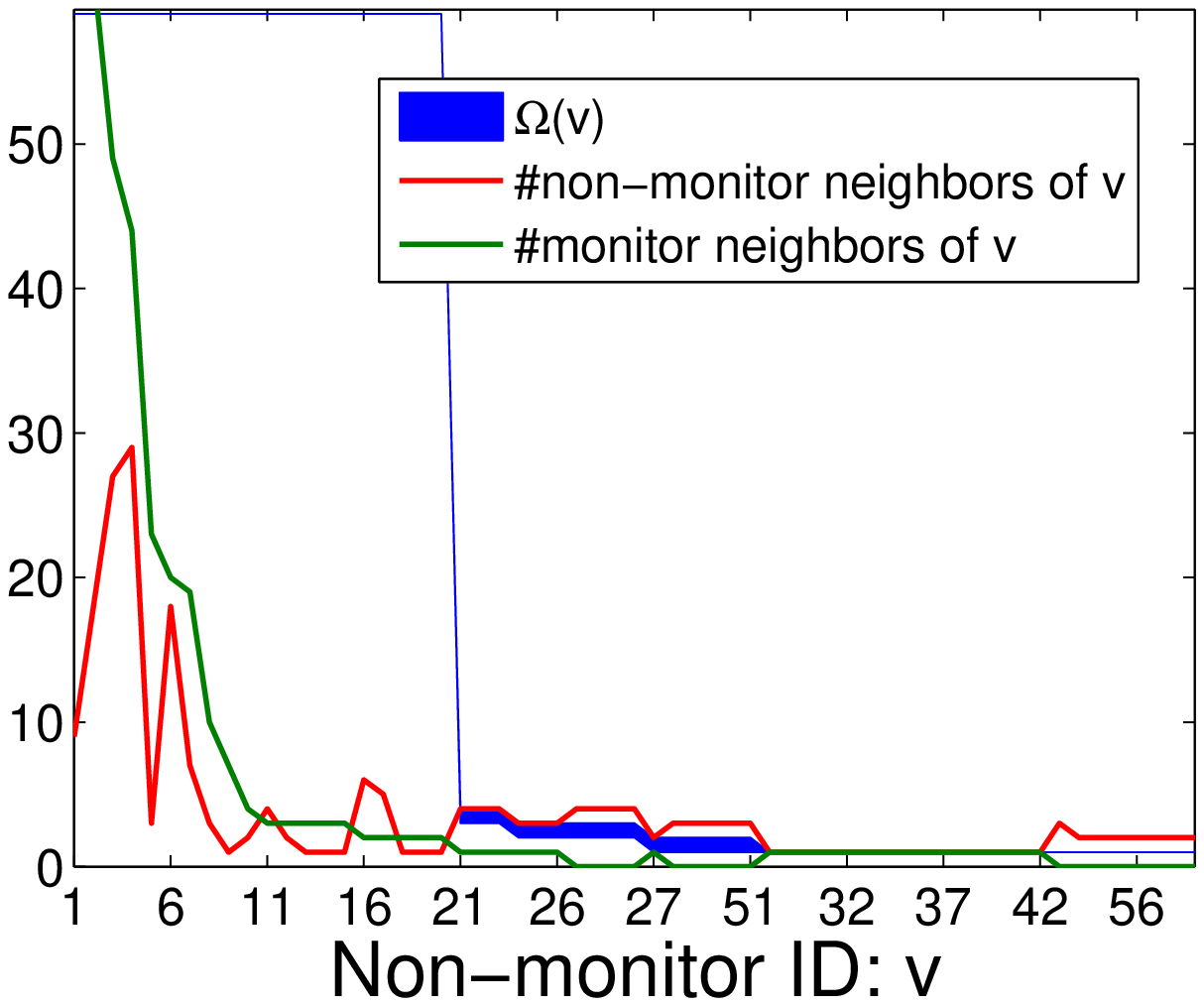}}
  \vspace{-.5em}\centerline{\small (b) Under CSP}
\end{minipage}\hfill
\begin{minipage}{.3\linewidth}
  \centerline{\includegraphics[width=1\columnwidth]{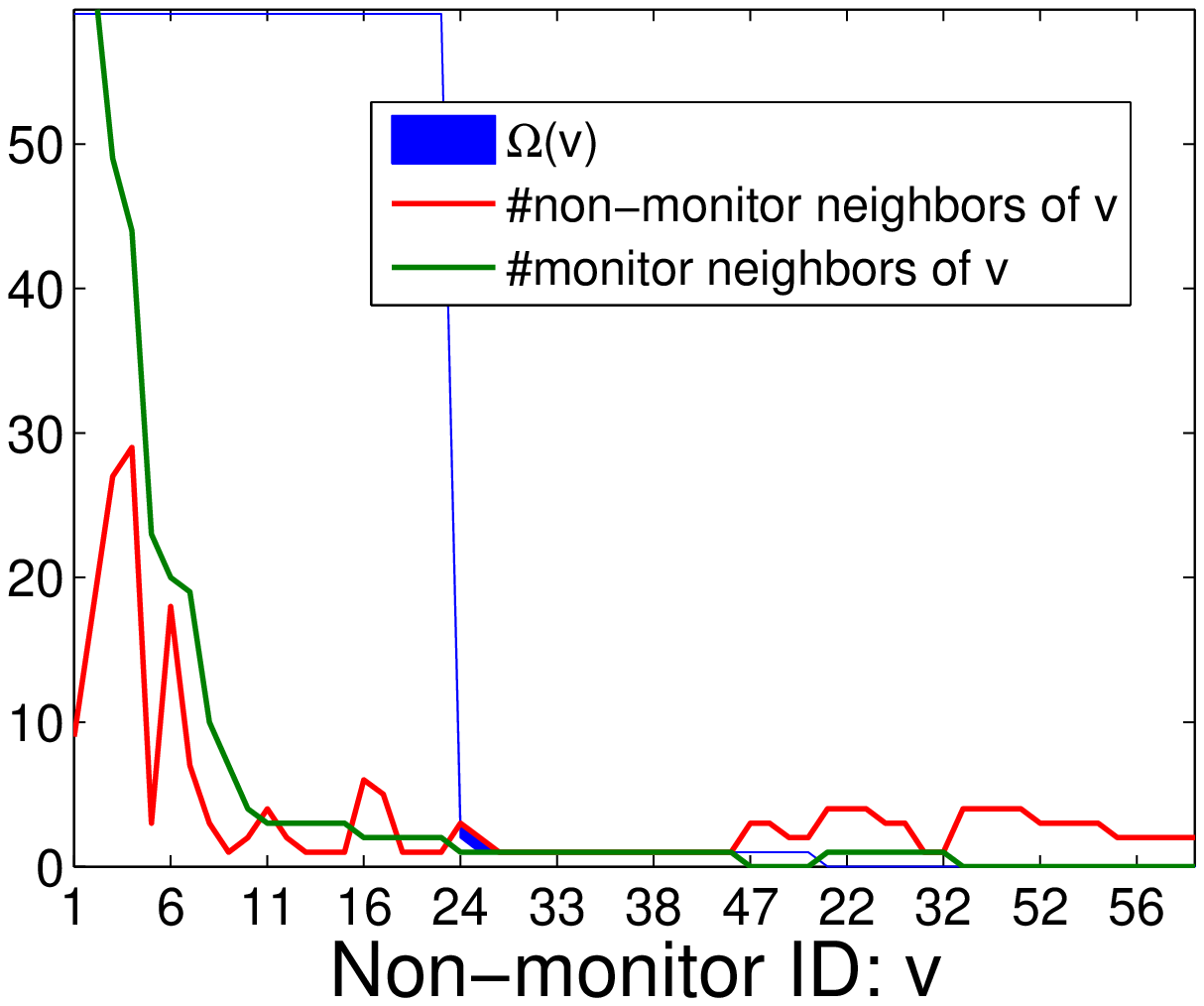}}
  \vspace{-.5em}\centerline{\small (c) Under UP}
\end{minipage}
\vspace{-.0em}
\caption{Node maximum identifiability index $\Omega(v)$ of CAIDA under different probing mechanisms ($|V|=355$, $|L|=483$, $\mu=296$). } \label{fig:node_maximum_identifiability_CAIDA}
\vspace{-.5em}
\end{figure*}

Next, we examine specific values of $\Omega(v)$ for each non-monitor $v\in N$ for selected instances of network topology and monitor placement. Our goal is to compare these values with node degrees to understand the correlation between the proposed identifiability measure and typical graph-theoretic node properties. Specifically,
we sort non-monitors in a non-increasing order of $\Omega(v)$ under each of the three probing mechanisms, and compare $\Omega(v)$ with the degrees of $v$ among monitors/non-monitors\footnote{Note that node IDs are different under different probing mechanisms due to the different order of $\Omega(v)$ values.}; see results in Fig.~\ref{fig:node_maximum_identifiability} for random topologies and in Fig.~\ref{fig:node_maximum_identifiability_Rocketfuel}--\ref{fig:node_maximum_identifiability_CAIDA} for AS topologies.
The results show strong correlations between $\Omega(v)$ and the degree of $v$, denoted by d$(v)$. Specifically, denote the number of neighbors of $v$ that are monitors by d$^m(v)$ and the number of neighbors of $v$ that are non-monitors by d$^n(v)$; the overall degree d$(v) = \mbox{d}^m(v)+\mbox{d}^n(v)$. If node $v$ has sufficient monitor neighbors (d$^m(v)\geq 1$ for CAP, d$^m(v)\geq 2$ for CSP), then $v$ is directly measurable and thus $\Omega(v) = \sigma$ regardless of the actual degree of $v$; if node $v$ does not have a sufficient number of monitors as neighbors, then $\Omega(v)\leq $d$(v)$ because if all neighbors of $v$ fail, then the state of $v$ cannot be determined by path measurements.
However, in the latter case, d$(v)$ is only a loose upper bound, and the exact value of $\Omega(v)$ depends on the overall topology, the locations of monitors, and the constraints on measurement paths. In this regard, our result can also be viewed as defining a new node property ($\Omega(v)$) that takes into account all these parameters.

\vspace{.5em}
Overall, we observe that CAP-type probing is hugely advantageous in uniquely monitoring node states under failures, especially when there are multiple failures and the network is sparse. This implies that in the absence of deploying monitors at every node, implementing controllable probing is an effective way to uniquely localize node failures. Our observation also stresses the importance of optimized monitor placement, especially when we are only interested in monitoring a subset of nodes, which is left to future work.


\section{Conclusion}\label{sec:Conclusion}

We studied the fundamental capability of a network in localizing failed nodes from binary measurements (normal/failed) of paths between monitors. We proposed two novel measures: \emph{maximum identifiability index} that quantifies the scale of uniquely localizable failures wrt a given node set, and \emph{maximum identifiable set} that quantifies the scope of unique localization under a given scale of failures. We showed that both measures are functions of the maximum identifiability index per node.
We studied these measures for three types of probing mechanisms that offer different controllability of probes and complexity of implementation. For each probing mechanism, we established necessary/sufficient conditions for unique failure localization based on network topology, placement of monitors, constraints on measurement paths, and scale of failures. We further showed that these conditions lead to tight upper/lower bounds on the maximum identifiability index, as well as inner/outer bounds on the maximum identifiable set. We showed that both the conditions and the bounds can be evaluated efficiently using polynomial-time algorithms. Our evaluations on random and real network topologies showed that probing mechanisms that allow monitors to control the routing of probes have significantly better capability to uniquely localize failures.

\bibliographystyle{IEEEtran}
\bibliography{mybibSimplifiedA,mybibSimplifiedB,mybibSimplifiedC}

\end{document}